\documentclass[journal]{IEEEtran}

\usepackage{amsfonts}
\usepackage{amsmath}
\usepackage{amsthm}
\usepackage{algorithm}
\usepackage{algorithmic}
\usepackage{graphicx}
\usepackage{float}
\usepackage{color}
\usepackage{tabularx} 
\usepackage{cite} 

\usepackage{subfigure}

\newtheorem{theorem}{Theorem}
\newtheorem{definition}{Definition} 
\newtheorem{corollary}{Corollary} 

\ifCLASSINFOpdf
 
\else
 
\fi

\begin{document}

\title{Security Analysis of Thumbnail-Preserving Image Encryption and a New Framework}
\author{Dong~Xie,~
        Zhiyang Li,~
        Shuangxi Guo,~
        Fulong~Chen,~
        and~Peng~Hu
\thanks{This paper is partially supported by the National Natural Science Foundation of China (Grant Nos. 61801004 and 61972438), the Natural Science
Foundation of Anhui Province of China (Grant Nos. 2108085MF219 and
2108085MF206), and the Key Research and Development Projects in Anhui
Province (Grant Nos. 202004a05020002 and 2022a05020049).}
\thanks{Dong Xie, Zhiyang Li, Shuangxi Guo, Fulong Chen, and Peng Hu 
are with the Anhui Provincial Key Laboratory of Network and Information
Security, and the School of Computer and Information, Anhui Normal
University, Wuhu 241002, China. (email: xiedong@ahnu.edu.cn)}
\thanks{Manuscript received Dec. XX, 2023; revised XXXX XX, 2024.}}

\markboth{Journal of \LaTeX\ Class Files,~Vol.~18, No.~9, September~2024}%
{Shell \MakeLowercase{\textit{et al.}}: Bare Demo of IEEEtran.cls for IEEE Journals}

\maketitle

\begin{abstract}
As a primary encryption primitive balancing the privacy and searchability of cloud storage images, thumbnail preserving encryption (TPE) enables users to quickly identify the privacy personal image on the cloud and request this image from the owner through a secure channel.
In this paper, we have found that two different plaintext images may produce the same thumbnail. It results in the failure of search strategy because the collision of thumbnail occurs. 
To address this serious security issues, we conduct an in-depth analysis on the collision probabilities of thumbnails, and then propose a new TPE framework, called \textit{multi-factor thumbnail preserving encryption} (MFTPE).
It starts from the collision probability of two blocks, extend to the probabilities of two images and ultimately to $N$ images.
Then, we in detail describe three specific MFTPE constructions preserving different combinations of factors, i.e., the sum and the geometric mean, the sum and the range, and the sum and the weighted mean.
The theoretical and experimental results demonstrate that the proposed MFTPE reduces the probability of thumbnails, exhibits strong robustness, and also effectively resists face detection and noise attacks.
\end{abstract}


\IEEEpeerreviewmaketitle

\section{Introduction}
\label{Introduction}

\subsection{Backgrounds}
\begin{figure*}
    \centering
    \includegraphics[width=0.9\textwidth]{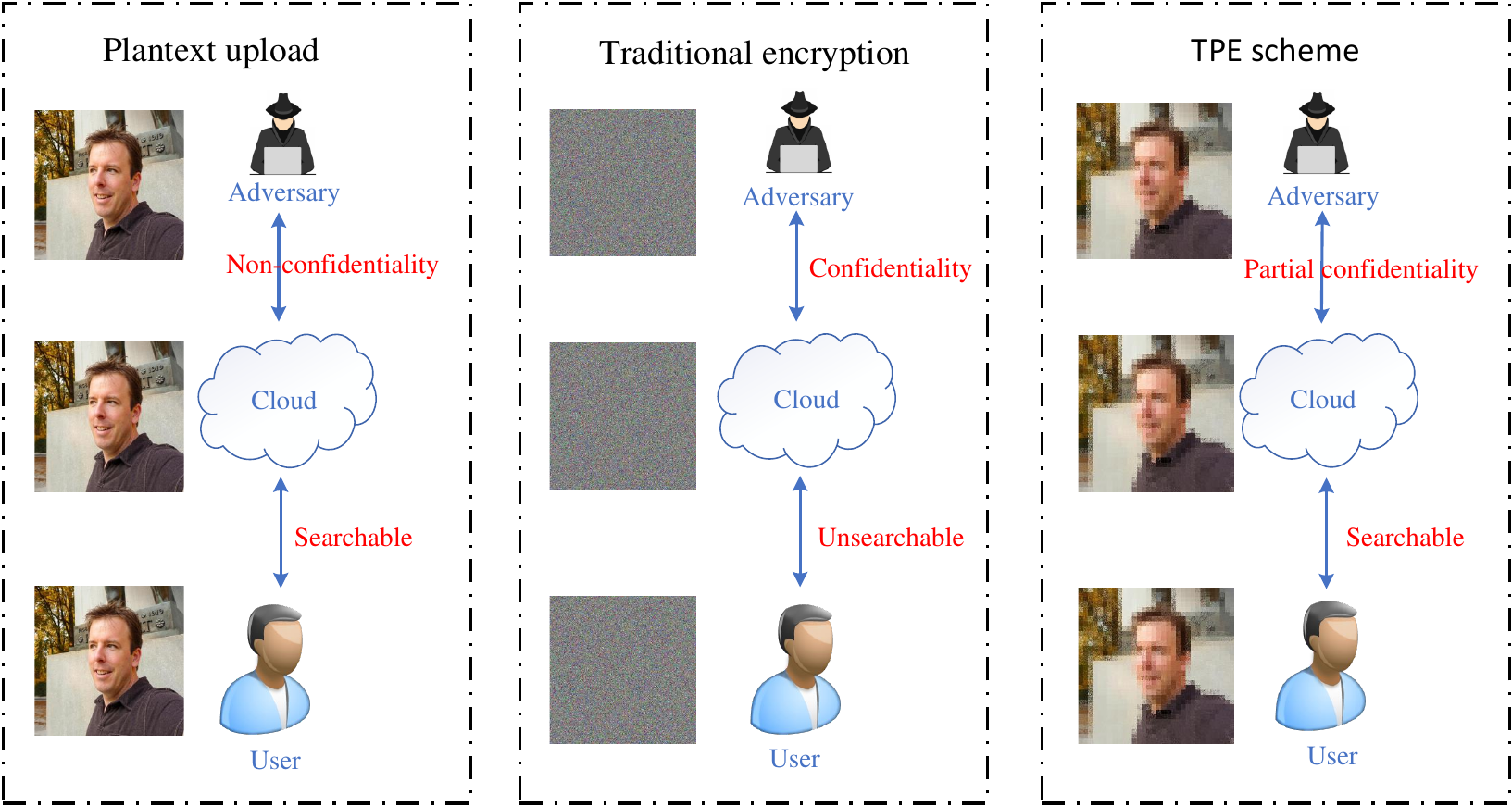}
    \caption{Comparison of different methods balancing the privacy and the serachability of could images.}
    \label{model}
\end{figure*}

\IEEEPARstart{T}{he} security of personal privacy images is a highly prominent concern in today's digital society. 
One of the most notable incidents involving privacy breaches was the Cambridge Analytica scandal in 2018, where the company illicitly obtained and exploited the personal data of tens of millions of Facebook users \cite{hinds2020wouldn}. 
Storing images in plaintext on cloud servers presents a significant vulnerability, as it can lead to the direct exposure of privacy if these servers become targets of malicious attacks. 
Therefore, the most common method is to encrypt the image before uploading it to the cloud.
Conventional image encryption methods can shield all the content in the plaintext image from unauthorized access by transforming the image into a noisy form. Regrettably, this approach proves inconvenient for users and is ill-suited for situations demanding swift identification of image content, such as image retrieval and image library management. 

Compared to plaintext uploading and traditional encryption, thumbnail preserving encryption (TPE) \cite{wright2015thumbnail,marohn2017approximate,tajik2019balancing,zhang2022f}, can balance the privacy and availability of cloud images (see Fig. 1).
The core idea of TPE is that it keeps the thumbnails of the plaintext and the ciphertext image unchanged. 
The thumbnail is a low-resolution form of the image, which is created by 
dividing the image into individual blocks with size $B \times B$ and subsequently converting the average of pixels within each block into a single pixel.
This means that even with the encrypted image, users can still gain some intuitive understanding of the image content through its thumbnail. This feature is particularly useful for the rapid recognition and browsing of images, especially in cases where a large number of images need to be quickly filtered or browsed.

Starting from the first TPE scheme \cite{wright2015thumbnail}, researchers mainly focus on improving the efficiency of the TPE schemes.
As far as we know, the earliest TPE scheme \cite{wright2015thumbnail} merely performed the permutation step of pixels within blocks, which resulted in high efficiency of the scheme.
But, its security deserve to be deliberated.
To enhance the security, a new optimal TPE scheme \cite{tajik2019balancing} has been proposed based on substitution-permutation framework.
The efficiency of this scheme was somewhat limited due to its grouping of two pixels.
Based on this framework, a significant improvement by increasing the number of encrypted pixels in a block each time was introduced \cite{zhang2022f}. 
Overall, there are many optimizations regarding the efficiency of TPE schemes \cite{xie2023}.
As far as we know, the theoretical security analysis of the existing TPE scheme are mainly based on Markov chain method \cite{tajik2019balancing}.

\subsection{Motivation}
For the security, existing TPE schemes analyze it either through experiments \cite{zhang2022high,chai2022preserving} or through the randomness of pixels within blocks based on Markov chains \cite{tajik2019balancing,zhang2022f,zhang2021hf}.
However, we have found that there exists a critical security issue, named \textit{thumbnail collision}, which means that two completely different images may produce the same thumbnail. 
For example, two images of girls have the same thumbnail (see Figure \ref{same_thumbnail}). 
Evidently, the collision probability of thumbnails of different images gradually increases as the number of cloud images increases.

\begin{figure}
    \centering
    \includegraphics[width=0.4\textwidth]{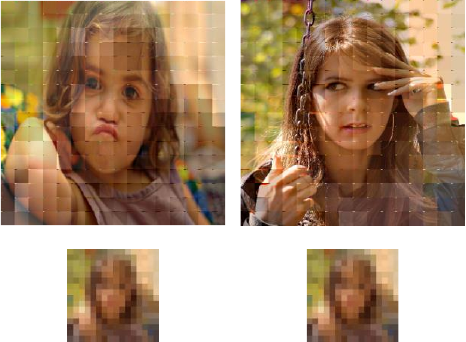}
    \caption{Different images with the same thumbnail}
    \label{same_thumbnail}
\end{figure}

In the TPE scenario, this thumbnail collision will cause \textit{serious consequences}.
Note that legitimate users can encrypt their privacy images and upload them to the cloud for economic or other purposes.
And other users can perform operations (e.g., image searching) on these encrypted images though thumbnails. 
For instance, when users search for their favorite music on some music software platforms (e.g., Spotity and Apple Music), the platforms will display that some music are free and some are charged.
For charged music, the platform will provide a few seconds of trial listening.
Note that these music snippets are akin to thumbnails of images. 
Suppose that two charged songs $S_{1}$ and $S_{2}$ have an identical preview snippet, i.e., a collision occurs, it is impossible for the user to distinguish them though this snippet.
Therefore, the user is not sure if he/she need to purchase this song because it is possible that he/she is willing to purchase $S_{1}$ instead of $S_{2}$.
Similarly, the situation of two cloud images with the same thumbnail can result in the failure of the search strategy.

\subsection{Contribution}

In this paper, we analyze the collision probability of thumbnails in existing TPE schemes, and provide a new framework, called multi-factor TPE (MFTPE), for addressing this collision issue. 
TPE schemes require the sum within blocks of plaintext and ciphertext to remain unchanged, but MFTPE schemes require multiple factors (e.g., the geometric mean, the range, and the weighted mean) to remain unchanged before and after encryption.
According to different circumstances, the user can choose one or more different factors to keep them unchanged.
By ensuring the preservation of several factors, it significantly reduces the collision probability of thumbnails.
The main contributions of this paper are summarized as follows.

\begin{itemize}
\item \textit{Analyzing the collision probability of TPE}. Starting with the collision probability of two blocks in images, we progressively proceed to investigate the collision probability of two images. Then we consider it 
to three images, and eventually to $N$ images.
It is imperative to emphasize that, as the number of images increases, the collision probability of thumbnails progressively increases. When the number of stored images reaches a certain threshold, the collision probability becomes a non-negligible security issue.
    \item \textit{Designing new encryption framework}. To reduce the collision probability, we propose a multiple factor TPE scheme, called MFTPE, for balancing the privacy and the searchability of cloud storage image. Under this framework, we proposed three specific MFTPE schemes based on different factor combinations, take the sum and the geometric mean of pixels as an example. The results indicate that the MFTPE scheme preserving the sum and the weighted mean exhibits the most significant fluctuations, and the encryption time is acceptable compared to existing TPE schemes.    
       \item \textit{Performing security analysis on MFTPE}. The theoretical results shows that the collision probability has been extremely reduced. The two different plaintext with the same thumbnail in TPE schemes have completely different thumbnails under the proposed MFTPE framework.   
    Also, the proposed MFTPE schemes can effectively resist noise attacks and face detection attacks. 
\end{itemize}

\subsection{Organization}
The remaining sections of this paper are outlined as follows.
In Section \ref{related work}, we provide a brief review of relevant works. Section \ref{Preliminaries} introduce the symbols used in this paper and the prelimilaries of TPE. 
In Section \ref{Security issues in TPE}, we step-by-step  analyze the collision probability of thumbnails in the traditional TPE schemes. 
In Section \ref{Proposed MFTPE Framework}, we propose a generalized TPE framework, referred to as MFTPE, and provide a detailed description of three specific schemes. 
Section \ref{Collision probability of MFTPE} gives the collision probality of MFTPE schemes, and Section \ref{experimnet} shows some experimental results. 
In Section \ref{comparison}, the comparison between TPE schemes and MFTPE schemes are introduced. 
Finally, Section \ref{conclusion} concludes this paper.

\section{Related work}\label{related work}

There exist two types of TPE schemes, i.e., approximate TPE and ideal TPE.
In approximate TPE schemes, the thumbnails of the plaintext and ciphertext images are not necessarily identical, and perfect decryption is not guaranteed. 
However, ideal TPE not only ensures that the thumbnails of the plaintext and ciphertext images are identical but also offers the perfect decryption of the ciphertext. 
Therefore, researchers have recently mainly focused on designing efficient and secure ideal TPE schemes.

\subsection{Approximate TPE}

The first approximate TPE scheme, called TPE-LSB and DRPE, were initially proposed by Marohn \textit{et al.} \cite{marohn2017approximate}, representing the first implementations of image encryption for JPEG images. TPE-LSB is characterized by its efficiency, although this comes at the cost of sacrificing the quality of decrypted images. 
DRPE, on the other hand, carries the risk of decryption failure.
Subsequently, Zhang \textit{et al.} \cite{zhang2021hf} proposed the HF-TPE scheme, wherein both the visual quality of the ciphertext image and the decrypted image asymptotically approaches those in the ideal TPE scheme. While it still cannot achieve completely lossless decryption, this scheme preserves more information and features of the images. 
Wang \textit{et al.} \cite{wang2023TPE} introduced the TPE-ISE scheme, which can control the compression ratios for ciphertext images. The perceived quality of ciphertext images are closer to plaintext images. Furthermore, TPE-ISE exhibits a strong resistance to noise, endowing the scheme with notable robustness.

In 2022, Ye \textit{et al.} \cite{ye2022tpe} designed a PRA-TPE scheme, combining reversible data hiding with an encryption scheme, achieving perfect recovery of encrypted images.
Furthermore, Ye \textit{et al.} \cite{ye2022noise} proposed the NF-TPE scheme, which employs a reversible data hiding scheme based on the prediction of the most significant bit. This scheme effectively eliminates noise and leverages local correlations between adjacent pixels to achieve perfect decryption of encrypted images.
In 2023, Ye \textit{et al.} \cite{ye2023usability} introduced a TPE scheme for the JPEG format images. This scheme initially transforms the plaintext into the frequency domain and subsequently applies frequency domain reversible data hiding and image encryption for achieving the lossless decryption.

\subsection{Ideal TPE}

The first ideal TPE scheme was proposed by Wright \textit{et al.}\cite{wright2015thumbnail}, employing a scrambling-only approach for image encryption. However, Jolfaei \textit{et al.} \cite{jolfaei2015security} pointed out the security vulnerabilities in this permutation-only scheme, i.e.,  the used permutation mapping can be recovered completely under the chosen-plaintext attack model.
To enhance the security of ideal TPE schemes, Tajik \textit{et al.}\cite{tajik2019balancing} introduced a framework of substitution-permutation networks for image encryption. They used Markov chain to model the randomness of pixels within a block for the security proof. However, this scheme performed the encryption phase with two pixels per group, which to some extent limited the connectivity of the Markov chain.

Building upon Tajik \textit{et al.}'s \cite{tajik2019balancing} work, Zhao \textit{et al.} \cite{zhao2021tpe2} proposed an encryption scheme using three pixels per group, which improved the connectivity of the Markov chain. In 2022, Zhang \textit{et al.} \cite{zhang2022f} introduced a multi-pixel encryption scheme allowed for encryption with pixel groups of arbitrary length, showcasing the scheme's flexibility while reducing encryption time and improving efficiency.

Another significant improvement was the integration of chaotic systems. The first scheme to combine chaotic systems with TPE was introduced by Zhang \textit{et al.} \cite{zhang2022high} in 2022. 
The scheme used the random chaotic matrices instead of pseudo-random functions, which greatly improve the efficiency. 
However, the randomness of chaotic systems could not be theoretically proven.
Chai \textit{et al.} \cite{chai2022preserving} pioneered the integration of genetic algorithms and TPE. They applied genetic algorithms for bit-level scramble and diffuse of pixels, introducing crossover and mutation operations. Additionally, they introduced a color histogram retrieval algorithm for ciphertext retrieval.

\section{Preliminaries}
\label{Preliminaries}
In this section, we will introduce the concept of thumbnail-preserving encryption and some cryptographic primitives. For simplicity, the symbols used in this paper are as shown in Table \ref{notation}.
\begin{table}
    \centering
    \caption{Notations}
    \label{notation}
    \begin{tabular}{cc}
    \hline
    Notation         & Description \\ \hline
    $\$$ & Represents a true random selector \\
    \((\mathbb{Z}_{d+1})^n\) & Set of vectors with $n$ elements in \(\mathbb{Z}_{d+1}\) \\
    $d$ &  Typically 255   \\
    \(\mathcal{M},M\)  & Plaintext space and plaintext image \\ 
    \(\mathcal{C},C\)  & Ciphertext space and ciphertext image \\ 
    \(K\) & Private key  \\ 
    \(W,H\) & Width and height  of the image \\
    \(B\) & Size of a block \\
    $\Phi$ & Arbitrary function defined on $M$ that represents a property \\
    \(\Phi^{-1}_d(s,n)\) & Set of vectors with sum $s$ and length $n$ \\
    \(\Psi_d(s,n)\) & Number of vectors in \(\Phi^{-1}_d(s,n)\) \\
    \(SN\) & Serial number \\
    \(SN^{*}\) & Serial number after encryption \\  
    \(sum(\vec{g})\) & Sum of elements of $\vec{g}$ \\  
    \(Thum(I)\) & Thumbnail of image $I$ \\  \hline
    \end{tabular}
\end{table}

\subsection{Nonce-Based Encryption}

TPE is a special type of nonce-based encryption \cite{tajik2019balancing}. Nonce-based encryption adds nonces to the encryption process, which makes multiple encryptions of the same plaintext produce different ciphertexts. For a plaintext $M$ and a nonce $T$, the encryption and decryption processes are as follows,
\begin{equation}
    \begin{aligned}
        &Enc_K(T,M) =C, \\
        &Dec_K(T,C) =M.
        \label{nonce based}
    \end{aligned}
\end{equation}

\subsection{Format-Preserving Encryption}
TPE is a specific instance of format-preserving encryption (FPE) \cite{tajik2019balancing,bellare2009format}. FPE refers to the encryption of plaintext with a particular format in such a way that the resulting ciphertext preserves the same format. 
For $M$, $T$, and $K$, we said that the scheme is $\Phi$-preserving if the following two conditions are satisfied,
 
\begin{equation}
    \begin{aligned}
        &Enc_K(T,M) \in \mathcal{M}, \\
        &\Phi(Enc_K(T,M))=\Phi(M).
    \end{aligned}
\end{equation}
 
Bellare \textit{et al.} introduced two definition for evaluating the security of FPE schemes, i.e., pseudorandom permutation (PRP) security and nonce-respecting (NR) security \cite{bellare2009format}.

\begin{definition}[PRP security \cite{bellare2009format}]
Let $F_\Phi$ denote the set of functions $F:\{0,1\}^{*} \times \mathcal{M} \rightarrow \mathcal{M}$ that satisfies $\Phi(F(T,M))=\Phi(M)$.
For all probabilistic polynomial time (PPT) oracle adversacies $\mathcal{A}$ and all $T$ and $M$, the FPE scheme has PRP security if
\begin{equation}
    \left|\mathop{Pr}_{K \leftarrow\{0,1\}^{\lambda}}\left[\mathcal{A}^{\operatorname{Enc}_{K}(\cdot, \cdot)}(\lambda)=1\right]-\mathop{Pr}_{F \leftarrow \mathcal{F}_{\Phi}}\left[\mathcal{A}^{F(\cdot, \cdot)}(\lambda)=1\right]\right|
\end{equation}
is negligible in the security parameter $\lambda$.
\end{definition}

\begin{definition}[Nonce-respecting \cite{tajik2019balancing}]
Let $\mathcal{A}$ be the PPT oracle machine, where its oracle takes two arguments. If it never makes two oracle calls with the same first argument, then $\mathcal{A}$ is called nonce-respecting.
\end{definition}

If an FPE scheme satisfies Def. 2 but only with respect to nonce-respecting distinguishers, then the scheme has NR security.

\subsection{Thumbnail-Preserving Encryption}
Generally, plaintext images and ciphertext images in TPE schemes have the same dimensions and thumbnails.
The method for generating thumbnails involves dividing the image into multiple blocks and then calculating the average values within each block. This average value becomes the pixel value for the thumbnail. The specific steps for generating thumbnails are as follows.
\begin{itemize}
    \item The plaintext image $M$ should be divided into three channels first. The pixel values of each channel are in the range \([0, d]\).
    \item Suppose that the dimension of each channel is $W \times H$. It is further divided each channel into smaller blocks of size $B \times B$. Without loss of generality, $B$ is a common factor of $W$ and $H$. 
    Consequently, each channel is partitioned into ${WH}/{B^2}$ blocks. 
    \item The pixel average of each block in the image serves as a pixel of the thumbnail.
\end{itemize}

Suppose that the sum of pixels within each block remains unchanged during the encryption process, then the thumbnail remains unchanged before and after the encryption phase.

\subsection{$\Phi_{d}^{-1}(s, n)$ and $\Psi_{d}(s, n)$}

The rank-then-encipher algorithm is the main method for constructing ideal TPE schemes \cite{tajik2019balancing,zhang2022f}.
It uses $\Phi_{d}^{-1}(s, n)$ to denote the set of vectors with length $n$ and sum $s$, where $d$ is the maximal value of pixels and typically $d=255$. Let $\vec{v}=(v_1,v_2,\cdots,v_n)$, $\Phi_{d}^{-1}(s, n)=\{\vec{v}| \sum  _{i=1}^n v_i=s\}$.
Similarly, 
$\Psi_{d}(s, n)$ denotes the number of vectors in $\Phi_{d}^{-1}(s, n)$.
Namely, $\Psi_{d}(s, n)=|\{v| \sum  _{i=1}^n v_i=s\}|$.
Evidently, if $n=1$ and $s<d$, $\Psi_{d}(s, 1)=1$.

When $n>1$ and $0 \leq s \leq d$,  $\Psi_{d}(s, n)=\Psi_{d}(s-i, n-1)$ if $v_1=i$. Since $i \in [0,s]$, thus we have 
 
\begin{equation}
    \Psi_{d}(s, n)= \sum_{i=0}^{s} \Psi_{d}(s-i, n-1).
    \label{summary of s-i}
\end{equation}

For example, $\Psi_{d}(2, 3) = \sum  _{i=0}^2 \Psi_{d}(2-i, 2)=\Psi_{d}(2, 2)+\Psi_{d}(1, 2)+\Psi_{d}(0, 2)$. Note that $\Psi_{d}(2, 2)=\sum _{i=0}^2 \Psi_{d}(2-i, 1)=\Psi_{d}(2, 1)+\Psi_{d}(1, 1)+\Psi_{d}(0, 1)=3$. Similarly, $\Psi_{d}(1, 2)=2$ and $\Psi_{d}(0, 2)=1$. Thus, $\Psi_{d}(2, 3) =6$. 
In fact, all vectors with sum 2 and length 3 is $(0,0,2)$, $(2,0,0)$,  $(0,2,0)$, $(0,1,1)$, $(1,0,1)$, $(1,1,0)$.

When $n>1$ and $d \leq s \leq dn$, the first element of the vector can be in $\left \{ l,\dots ,d \right \}$ since the sum of the elements is $s$, where $l=\max \{0,s-(n-1)\times d\} $.
If $s<(n-1)d$, $l=0$. If $s>(n-1)d$, then the first element of the vector is non-zero. Therefore, 
\begin{equation}
     \Psi_{d}(s, n)=\sum_{i=0}^{d} \Psi_{d}(s-i, n-1).
    \label{phi 0 to d}
\end{equation}

Generally, the recursive formula of $\Psi_{d}(s, n)$ are as follows.
\begin{equation}
    \begin{array}{l}
    \Psi_{d}(s, n)= 
    \left\{\begin{array}{ll}
    1 & n=1 ~\& ~ 0 \leq s \leq d, \\
    \sum_{i=0}^{s} \Psi_{d}(s-i, n-1) & n>1 ~\& ~0 \leq s \leq d, \\
    \sum_{i=0}^{d} \Psi_{d}(s-i, n-1) & n>1 ~\&~ d \leq s \leq dn, \\
    0 & \text { otherwise.}
    \end{array}\right.
    \end{array}
    \label{vector set size}
\end{equation}

\section{Collision probability of thumbnails in TPE}\label{Security issues in TPE}


\subsection{Multiple random blocks}

\begin{theorem}
Suppose that there exist a block with \(n\) pixels $B_1$ and $s_1=sum(B_1)$, where $s_1$ represent the sum of pixels in $B_1$. Then, for a random block $B_2$ with \(n\) pixels, the collision probability of $B_1$ and $B_2$ is 
\begin{equation}
    Pr[sum(B_1)=sum(B_2)]=\frac{\Psi_d(s_1,n)}{(d+1)^{n}}.
\end{equation}
\end{theorem}

\begin{proof}
Given a block  $B_1$  comprising $n$ pixels, the sum of pixels can range from $0$ to $dn$.
Note that $\Psi_d(s_1,n)$ represent the number of vectors with length $n$ and sum $s_1$.
Thus, the collision probability of $B_1$ and $B_2$ is $\Psi_d(s_1,n)/(d+1)^n$.
\end{proof}

\textit{Remark}. When $s \approx \lfloor \frac{n\times d}{2}\rfloor$, $\Psi_d(s,n)$ reaches its maximum value. Therefore, the collision probability is also at its peak. 

\begin{theorem}
Suppose that there exist two random blocks $B_1$ and $B_2$, and each consists of \(n\) pixels. Then, the collision probability of these two blocks, i.e., the probability of $sum(B_1)=sum(B_2)$, is
\begin{equation}
    Pr[sum(B_1)=sum(B_2)]=\frac{2+ \sum  _{s=1}^{dn-1}A^2_{\Psi_d(s,n)}}{(d+1)^{2n}}.
\end{equation}
\end{theorem}

\begin{proof}
Let $B_1=(v_{11},v_{12},\cdots,v_{1n})$ and $B_2=(v_{21},v_{22},\cdots,v_{2n})$, where $v_{ij} \in [0,d]$ for $i \in \{1,2\}$ and $j \in [1,n]$. We have
    \begin{align}
			&Pr[sum(B_1)=sum(B_2)]= \notag \\
			&\frac{\sum _{s=0}^{dn}Pr[sum(B_1)=sum(B_2)=s]}{(d+1)^{2n}} 
    \end{align}
Let $s=sum(B_1)=sum(B_2)$.
When $s=0$, there is only one case, i.e., $B_1=B_2=(0,0,\cdots,0)$.
Similarly, there is one case if $s=dn$.
When $1<s<dn$, there exists $\Psi_d(s,n)$ vectors with length $n$ such that the sum is $s$. Therefore, the event $sum(B_1)=sum(B_2)=s$ has $A^2_{\Psi_d(s,n)}$ cases for $1<s<dn$.
This completes the proof.
\end{proof}

\begin{corollary}
Suppose that there exist three random blocks $B_1$, $B_2$ and $B_3$, and each consists of \(n\) pixels. Then, the collision probability of these three blocks is
\begin{equation}
    Pr[sum(B_i)=sum(B_j)]^{i \neq j}_{i,j\in \{1,2,3\}}=\frac{2+ \sum  _{s=1}^{dn-1}A^3_{\Psi_d(s,n)}}{(d+1)^{2n}}.
\end{equation}
\end{corollary}

\subsection{Two random images}

Without loss of generally, we only consider the case where the image has only one channel in the following theorems. The theoretical results can be obtained similarly if the image has multiple channels.

\begin{theorem}
Suppose that there exists an image $I_1$ with \(m\) blocks, and each block consists of \(n\) pixels. Let $s_{1i}$ represent the sum of the $i$-th block of $I_1$.
The collision probability between $I_1$ and $I_2$, i.e., $Thum(I_1)=Thum(I_2)$, is
\begin{equation}
    Pr[Thum(I_1)=Thum(I_2)]=\prod_{i=1}^m \frac{\Psi_d(s_{1i},n)}{(d+1)^{n}}. 
\end{equation}
\end{theorem}

\begin{proof}
Let $B_{1i}$ and $B_{2i}$ denote the $i$-th block of $I_1$ and $I_2$, where $1\leq i\leq m$. For $Thum(I_1)=Thum(I_2)$, it has $sum(B_{1i})=sum(B_{2i})$ for any $i\in \{1,2,\cdots,m\}$. According to Theorem 1, we have this theorem.
\end{proof}

Based on the proof methods of Theorem 2 and Theorem 3, the following theorem can be obtained.

\begin{theorem}
Suppose that there exist two random images $I_1$ and $I_2$. Each consists of $m$ blocks and each block has \(n\) pixels. Then, the collision probability of these two images, i.e., the probability of $Thum(I_1)=Thum(I_2)$, is
\begin{equation}
    Pr[Thum(I_1)=Thum(I_2)]=\frac{(2+ \sum  _{s=1}^{dn-1}A^2_{\Psi_d(s,n)})^m}{(d+1)^{2mn}}.
\end{equation}
\end{theorem}

\subsection{Three random images}
\begin{theorem}
    Assume that there exists three images $I_1$, $I_2$, and $I_3$. Each image has \(m\) blocks, and each block consists of \(n\) pixels.
   For $k\in \{1,2\}$ and $i\in \{1,2,\cdots,m\}$, let $s_{ki}=sum(B_{ki})$, where $B_{ki}$ denotes $i$-th block of the image $I_k$. 
    Then the collision probability between the thumbnails of three images, i.e., $Thum(I_i)=Thum(I_j)$ with $i \neq j$ and $i,j \in \{1,2,3\}$, is 
\begin{align}
    &Pr[Thum(I_i)=Thum(I_j)]^{i \neq j}_{i,j \in \{1,2,3\}}\notag \\
   &=\frac{\prod_{i=1}^m\Psi_d(s_{1i},n)+\prod_{k=1}^2\prod_{i=1}^m\Psi_d(s_{ki},n)}{(d+1)^{mn}}\notag \\
&-\frac{\prod_{i=1}^m\Psi_d^{2}(s_{1i},n)\prod_{i=1}^m\Psi_d(s_{2i},n)}{(d+1)^{2mn}}.
\end{align}
\end{theorem}

\begin{proof}
It is slightly complicated to calculate the collision probability of thumbnails between any two images in three images. Because there exists forth cases, i.e., $Thum(I_1) = Thum(I_2)$, $Thum(I_1) = Thum(I_3)$, $Thum(I_2) = Thum(I_3)$, and $Thum(I_1) = Thum(I_2)=Thum(I_3)$. Therefore, we adapt the proof method which is similar to the birthday attack route of hash functions. We have
\begin{align}
 &Pr[Thum(I_i)=Thum(I_j)]^{i \neq j}_{i,j \in \{1,2,3\}} \notag\\
&=1 - Pr[Thum(I_1)\neq Thum(I_2) \neq Thum(I_3)].
\end{align}

From Theorem 3, we have
\begin{equation}
Pr[Thum(I_1)\neq Thum(I_2)]=1-\frac{\prod_{i=1}^m\Psi_d(s_{1i},n)}{(d+1)^{mn}}.
\end{equation}

Similarly, given $I_1$ and $I_2$, the probability that the thumbnail of $I_3$ is different from both $I_1$ and $I_2$ is
\begin{align}
Pr[(Thum(I_3)&\neq Thum(I_1)) \land (Thum(I_3)\neq Thum(I_2))] \notag\\
&=1-\frac{\prod_{k=1}^2\prod_{i=1}^m\Psi_d(s_{ki},n)}{(d+1)^{mn}}.
\end{align}

From Eqs. (15) and (16), we have 
\begin{align}
&Pr[Thum(I_1)\neq Thum(I_2) \neq Thum(I_3)] \notag \\
&=Pr[Thum(I_1)\neq Thum(I_2)]\times \notag \\
&Pr[(Thum(I_3) \neq Thum(I_1)) \land (Thum(I_3)\neq Thum(I_2))]\notag \\
&=(1-\frac{\prod_{i=1}^m\Psi_d(s_{1i},n)}{(d+1)^{mn}})(1-\frac{\prod_{k=1}^2\prod_{i=1}^m\Psi_d(s_{ki},n)}{(d+1)^{mn}})\notag \\
&=1-\frac{\prod_{i=1}^m\Psi_d(s_{1i},n)+\prod_{k=1}^2\prod_{i=1}^m\Psi_d(s_{ki},n)}{(d+1)^{mn}}\notag \\
&+\frac{\prod_{i=1}^m\Psi_d^{2}(s_{1i},n)\prod_{i=1}^m\Psi_d(s_{2i},n)}{(d+1)^{2mn}}.
\end{align}

The theorem can be proved based on Eqs. (14) and (17).
\end{proof}

\begin{figure*}
    \centering
    \includegraphics[width=1\textwidth]{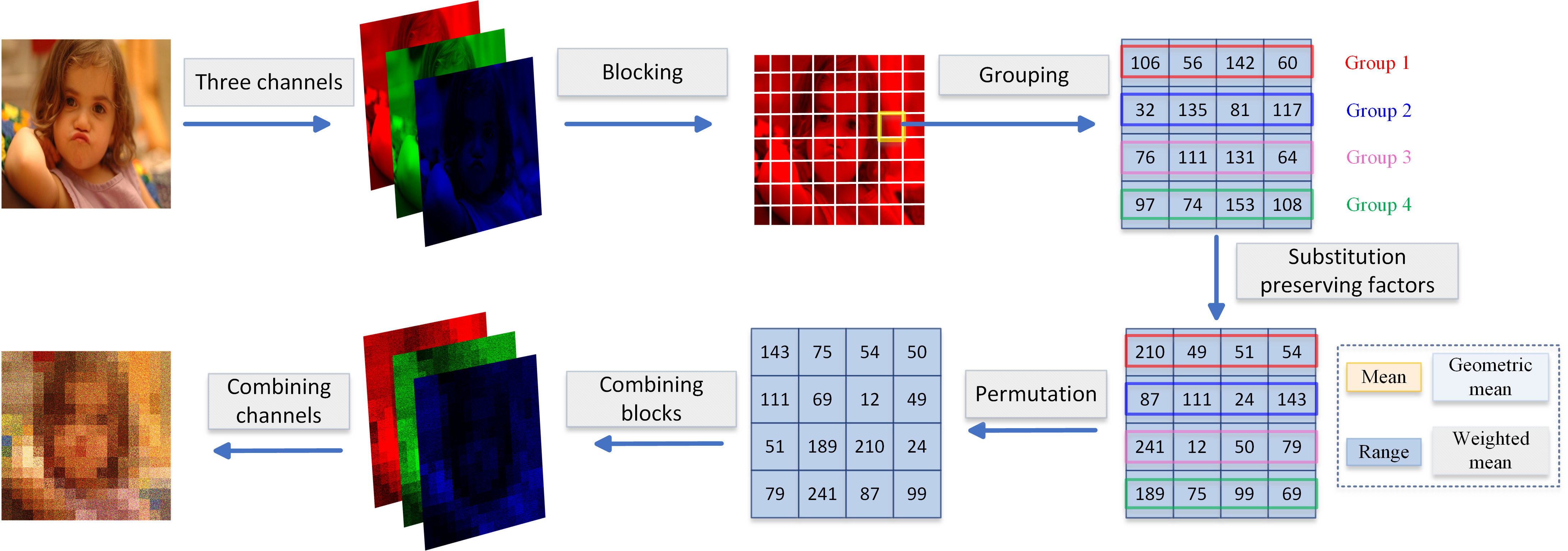}
    \caption{The encryption process of the MFTPE framework.}
    \label{encryption process}
\end{figure*}

\subsection{N random images}
If we generalize the result of Theorem 5 to any $N$ images, we have the following theorem.
\begin{theorem}
 Assume that there exists $N$ images $\{I_k\}_{k=1}^{N}$. Each image has \(m\) blocks, and each block consists of \(n\) pixels.
For $k\in \{1,2,\cdots,N\}$ and $i\in \{1,2,\cdots,m\}$, let $s_{ki}=sum(B_{ki})$, where $B_{ki}$ denotes $i$-th block of the image $I_k$. Let $S_1=\prod_{i=1}^m\Psi_d(s_{1i},n)$, $S_2=\prod_{i=1}^m\Psi_d(s_{2i},n)$, and $S_N=\prod_{i=1}^m\Psi_d(s_{Ni},n)$.
Then the collision probability between the thumbnails of these $N$ images is 

\begin{align}
    &Pr[Thum(I_i)=Thum(I_j)]^{i \neq j}_{i,j \in \{1,2,3\}}\notag \\
    &=1-\prod_{i=1}^{N}(1-\frac{\prod_{j=1}^{i}S_j}{(d+1)^{mn}}).
\end{align}
\end{theorem}

\section{Proposed MFTPE Framework}
\label{Proposed MFTPE Framework}

\subsection{Overview}
Each channel of plaintext images can be viewed as two dimensional matrices, whose elements belong to \(\mathbb{Z}_{d+1}\), typically with \(d\) being 255.
There exists a private key \(K \stackrel{\$}{\leftarrow}\{0,1\}^{\lambda}\), where  \(\lambda\) represents the security parameter of the system. 
The proposed framework also need a secure pseudorandom function (PRF) \cite{goldreich1986construct}, whose output sequences can not be distinguished from true random sequences.

\subsection{Encryption}

The proposed MFTPE framework (Fig. \ref{encryption process}) mainly consists three steps, i.e., partition, substitution, and permutation.

\textit{Partition}. The scheme first divide the plaintext image into multiple channels and then further divide these channels into individual blocks with size \(B\times B\). 
As mentioned by Tajik \textit{et al.} \cite{tajik2019balancing}, it is quite a challenging task to preserve the sum of pixels within an entire block once-only.
Therefore, it needs to group the pixels within a block in such a way that the sum of the pixels within a group remain unchanged.

\textit{Substitution}.
The substitution steps of almost all ideal TPE schemes  \cite{tajik2019balancing,zhang2022f} are based on the rank-then-encipher approach \cite{bellare2009format}. 
 Given a vector $\Vec{g}$ with sum $s$ and length $n$, the rank function is employed to obtain a serial number $SN$. Note that $SN$ is an integer and belongs to $[0, \Psi_d(s,n)-1]$ \cite{tajik2019balancing}.
Then the serial number $SN$ is encrypted using the secret key $K$ and obtain a new serial number $SN^{*}$. After that, $SN^{*}$ is mapped to a new vector $\Vec{g^{*}}$ within $\Phi_d^{-1}(s,n)$ using the $rank^{-1}$ function. Note that the sums of elements in $\Vec{g}$ and $\Vec{g^{*}}$ are the same, i.e., $sum(\Vec{g})=sum(\Vec{g^{*}})$.

Different to existing rank-then-encipher approach \cite{bellare2009format}, The substitution step of the proposed MFTPE framework can flexibly set multiple values of group pixels to remain unchanged, such as the range, the geometric mean, and the weighted mean.

\textit{Permutation}.
After the substitution phase, it is necessary to permutation the entire block, takeing the 
Fisher-Yates shuffle algorithm for example.
Note that the operation of substitution and permutation is called an encryption round, and an image can be encrypted multiple rounds for security.

In fact, we can select different combinations of factors to remain unchanged before and after encryption.
For better understanding, here we choose three combinations to introduce the substitution step in detail.
\begin{algorithm}[htb]
    \caption{Rank function preserving sum and geometric mean}
    \begin{algorithmic}[1] 
        \REQUIRE $\Vec{g}=\{a,b,c\}$
        \ENSURE $SN$, $\Phi_d^{-1}(s,p,n)$ 
        \STATE {$s=\sum \Vec{g}=a+b+c$}
        \STATE {$p = \prod \Vec{g}=abc$}
        \STATE {$count = 0$}
        \FOR {$ i = 0$ to $i < \min(s+1,256)$ }
            \FOR{$j=0$ to $j<\min(s+1,256)$}
                \STATE {$k = s-i-j$}
                \IF{$0 \leq k \leq 255$}
                    \IF{ $i*j*k =p$}
                        \IF{$i=a$ and $j=b$ and $k=c$}
                            \STATE {$SN = count$}                            
                        \ENDIF
                        \STATE {\textbf{Store} $SN$, $\Phi_d^{-1}(s,p,n).push(i,j,k)$}
                        \STATE {$count$++}
                    \ENDIF
                \ENDIF
            \ENDFOR
        \ENDFOR        
    \end{algorithmic}
\end{algorithm}

(1) \textit{The sum and the geometric mean}

The aim of this combination is to maintain the consistency of the sum and the geometric mean (i,e,, the product) of pixels when performing the substitution step.
Given a vector $\Vec{g}=(g_1,g_2, \cdots, g_n)$, let 
\begin{equation}
    \Phi_d^{-1}(s,p,n)=\{\Vec{g}|(\sum _{i=1}^{n} g_i=s) \land (\prod  _{i=1}^{n} g_i=p)\}.
\end{equation}
Also, we similarly use $\Psi_d(s,p,n)$ to represent the size of the set $\Phi_d^{-1}(s,p,n)$.
Note that if $n=2$, it is meaningless to guarantee both the sum and the product at the same time. Let $\vec{g}=(a,b)$, the vector after substitution must be $\vec{g}=(b,a)$.
Therefore, we introduce the substitution phase with groups of three pixels here.

Algorithm 1 shows the rank function when preserving the sum and the geometric mean of pixels.
Taking a pixel group \(\Vec{g}=(a,b,c)\) as an example, the rank function is employed to determine the serial number $SN$, which belongs to $[0,\Psi_d(a+b+c,abc,3)-1]$.
Similar to \cite{tajik2019balancing}, $SN$ should be encrypted to get a new number $SN^{*}$ using the private key $K$, 
\begin{align}
SN^{*} &=(SN+K) \bmod \Psi(s,p,n). \label{geometric re} 
\end{align}

Generally, the $rank^{-1}$ function is used for mapping $SN^{*}$ to a new vector $\vec{g^{*}}$, which belongs to $\Phi^{-1}_d(s,p,n)$ \cite{tajik2019balancing}. 
However, the corresponding relationship between the serial number and the vector has already been recorded (see Alg. 1) in the proposed $rank$ function for improving the efficiency.
Therefore, it can quickly lookup the table based on $SN^{*}$ to obtain its corresponding vector $\vec{g^{*}}$.

\begin{figure}
    \centering
    \includegraphics[width=0.5\textwidth]{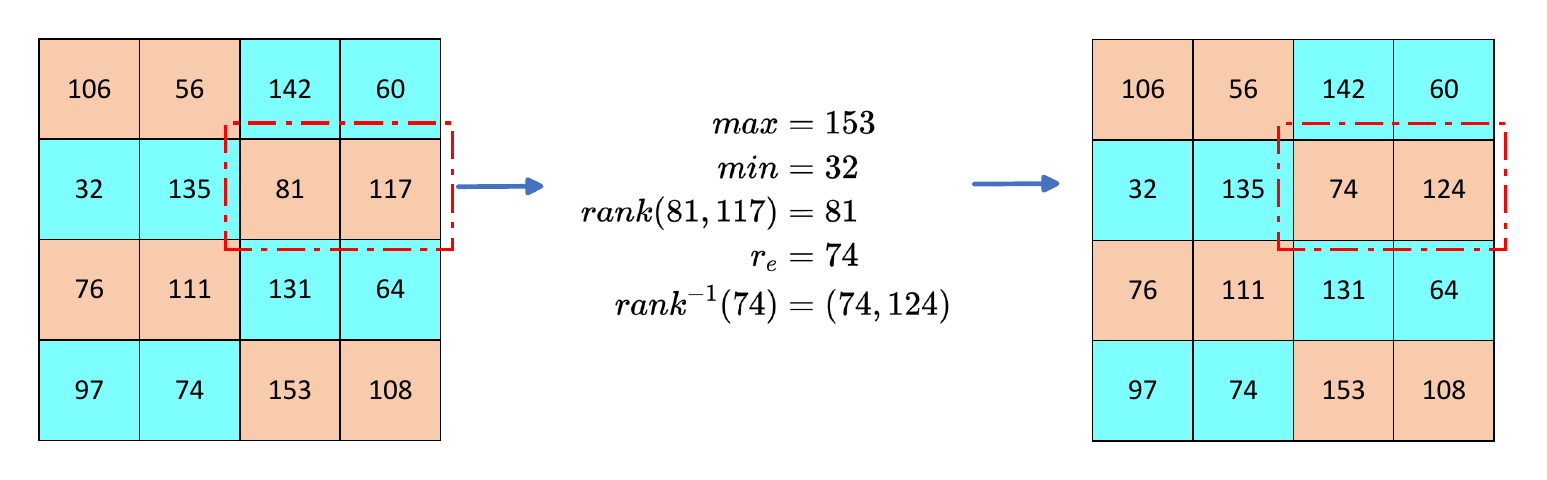}
    \caption{The substitution process of the MFTPE scheme when preserving the sum and the range.}
    \label{range process}
\end{figure}

(2) \textit{The sum and the range}
\label{pix range}

Given a block, the range is the difference between the maximum and minimum elements in this \textit{block}. The aim of this case is to maintain the consistency of the sum and the range of the block when performing the substitution step.
For simplicity, here we exemplify with a pixel group containing two pixels.
Given a vector $\Vec{g}=(a, b)$, the $rank$ function is 
\begin{align}
    rank(a,b) &=a. \label{range rank}
\end{align}
Note that the used $rank$ function is slightly different to that of Tajik \textit{et al.} scheme \cite{tajik2019balancing}.

$SN$ is obtained through the $rank$ function, i.e., $SN=a$. 
Let $max$ and $min$ respectively denote the maximum and minimum values within the block, and $s=a+b$.
Then, $SN^{*}$ can be computed by 

\begin{align}
    SN^{*} &=(SN-\alpha+K)\bmod(\beta-\alpha+1)+\alpha, \label{range re}    
\end{align}
where $\alpha$ and $\beta$ are 
\begin{align}
    \alpha &=\max(min+1,s-max+1), \label{alpha} \\
    \beta &= \min(s-min-1,max-1). \label{beta}
\end{align}

After that, the $rank^{-1}$ function is 
\begin{align}
   rank^{-1}(SN^{*}) &= (SN^{*},s-SN^{*}) \label{range unrank}
\end{align}

Fig. 4 shows the diagrammatic sketch of the proposed method. Note that if one pixel within the pixel group $\Vec{g}$ is either $max$ or $min$, the pixel positions are directly swapped without undergoing the following encryption phase. For instance, suppose that $max_{B}=100$ and $min_{B}=10$ for block $B$. If $\vec{g}=(100,12)$, then $\vec{g}^{*}=(12,100)$.

(3) \textit{The sum and the weighted mean}

Given a vector $\Vec{g}=\{g_1, g_2, \cdots, g_n\}$ and a weighted parameters $\Vec{w}=\{w_1, w_2, \cdots, w_n\}$, the weighted mean of $\Vec{g}$ with $\Vec{w}$ is defined as $w= (\sum_{i=1}^n w_ig_i)/(\sum_{i=1}^n w_i)$. 
Similarly to $\Phi^{-1}_d(s,p,n)$ , we also use $\Phi^{-1}_d(s,w,n)$ to denote the set of vectors 
whose sum is $s$ and the the weighted mean is $w$.
Due to the exceedingly small $\Phi^{-1}_d(s,w,n)$ value for pairs of two pixels, it is necessary to have at least three pixels within a group.

Firstly, the $rank$ function (see Alg. 2) is used to determine the serial number of $\Vec{g}$. 
Suppose that the pixel group $\Vec{g}$ is represented as $\Vec{g}=(a,b,c)$, and $w_1, w_2$, and $w_3$ are three random numbers. Here we state that it does not require that $\sum _{i=1}^3 w_i=1$.
Once $SN$ is obtained, it is encrypted using Eq. (\ref{wei re}). Finally, the function $rank^{-1}$ is used to map $SN^{*}$ back into $\Phi^{-1}_d(s,w,n)$.

\begin{align}
    SN^{*} &=(SN+K) \bmod \Psi(s,w,n) \label{wei re} 
\end{align}

\subsection{Decryption}
The decryption process involves performing the inverse of the encryption operations. The decrypter first divide the ciphertext image into multiple channels, and then participate each channel into multiple blocks. After that, it need to undergo the inverse permutation and the inverse substitution process. Finally, Combing blocks and channels to form the plaintext image.
Note that when the plaintext image performs multiple rounds of encryption, then the decryption phase should be perform the same round of decryption.

\section{Collision probability of MFTPE} \label{Collision probability of MFTPE}

In Section \ref{Security issues in TPE}, we have analyze the collision probability theoretically with existing TPE schemes.
The theoretical results show that the collision probabilities are directly proportional to the 
value of $\Psi_d(s,n)$.
For the proposed MFTPE schemes, the analysis method is the same, and the results are simply to change $\Psi_d(s,n)$ to $\Psi_d(f_1,f_2,\cdots, f_n,n)$ function, where $f_i$ denotes the $i$-th unchanged factor.
Evidently, given a specific $s$ and $n$, $\Psi_d(s,n)$ must be larger than $\Psi_d(s,p,n)$ and $\Psi_d(s,w,n)$.
Therefore, the collision probabilities of MFTPE schemes are lower than those of TPE schemes.

\begin{algorithm}[H]
    \caption{Rank function preserving the sum and the weighted mean}
    \begin{algorithmic}[1] 
        \REQUIRE $\Vec{g}=\{g_1,g_2, g_3\}$, $\Vec{w}=\{w_1,w_2,w_3\}$
        \ENSURE $SN$, $\Phi_d^{-1}(s,w,n)$ 
        \STATE {s = $\sum _{i=1}^{n}{g_i}$}
        \STATE $w = (\sum_{i=1}^3w_ig_i)/(\sum_{i=1}^3w_i)$ 
        \STATE {$count = 0$}
        \FOR {$ i = 0$ to $i < \min(s+1,256)$ }
            \FOR{$j=0$ to $j<\min(s+1,256)$}
                \STATE $k = s-i-j$
                \IF{$0 \leq k \leq 255$}
                    \IF{$i\times w_1+j\times w_2+k\times w_3 = w$}
                        \IF{$i=a$ and $j=b$ and $k=c$}
                            \STATE {$SN = count$}
                            \ENDIF
                        \STATE {\textbf{Store} $\Phi_d^{-1}(s,w,n).push(i,j,k)$}
                        \STATE {$count$++}
                    \ENDIF
                \ENDIF
            \ENDFOR
        \ENDFOR        
    \end{algorithmic}
\end{algorithm}

For simplicity, we assume that there exists an image $I$ with one channel, and it just has two blocks. Each blocks has $6$ random pixels. 
The randomly generated pixels in the first block is $(25,175,163,254,51,58)$, and the pixels in the second block is $(18,199,87,85,204,173)$. If the length of the pixels group is 2, then the total sum of $\Psi_d$ values for this image is $930$. If the length of the pixels group is 3, then the value is $183366$.
However, if we used the MFTPE scheme preserving the sum and the range and the length of the pixels group is 2, the total sum of $\Psi_d$ is $642$, which is decreased by $30.9\%$ compared to TPE schemes. 
Suppose that we set the length of group pixels to 3.
For MFTPE scheme preserving the sum and the geometric mean, the value is $24$. For MFTPE scheme preserving the sum and the weighted mean, the value is $292$. 
This specific example also provides evidence for that the proposed MFTPE schemes has a lower collision probability than TPE schemes.

\begin{figure}[htbp]  
  \centering     
  \subfigure[original image]
  {
      \includegraphics[width=0.11\textwidth]{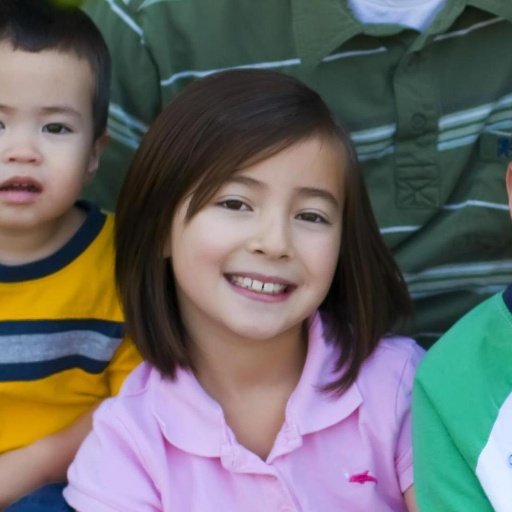}
  }\hspace{-3mm}
  \subfigure[$8 \times 8$]
  {
      \includegraphics[width=0.11\textwidth]{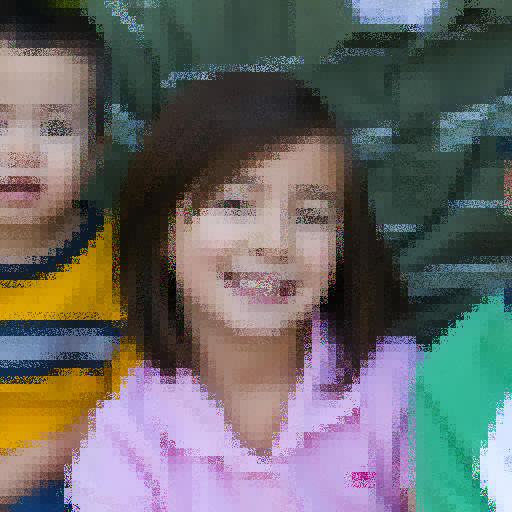}
  }\hspace{-3mm}
  \subfigure[$16 \times 16$]
  {
      
      \includegraphics[width=0.11\textwidth]{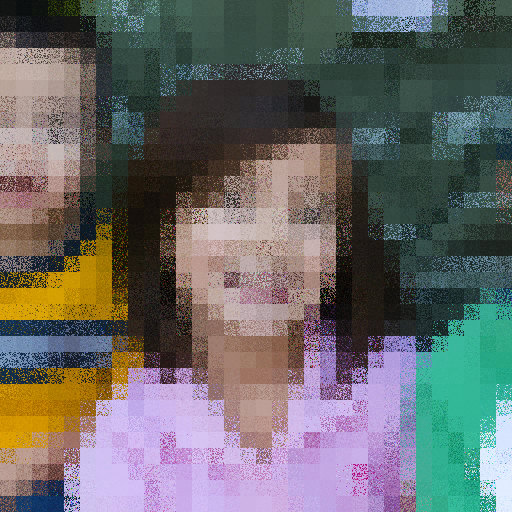}
  }\hspace{-3mm}
  \subfigure[$32 \times 32$]
  {
      
      \includegraphics[width=0.11\textwidth]{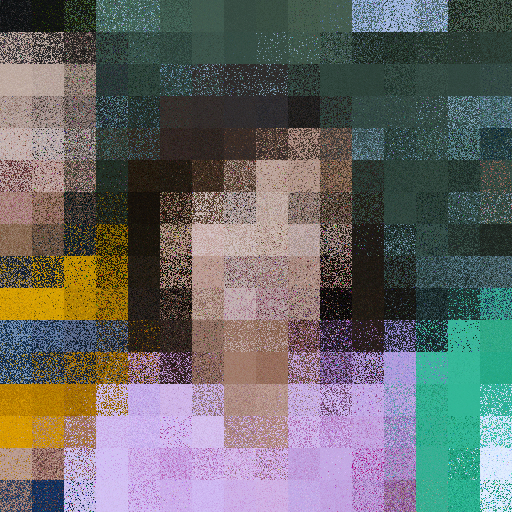}
  }\hspace{-3mm}
  \caption{The ciphertext images under the MFTPE framework preserving the sum and the geometric mean}
  \label{product}
\end{figure}

\begin{figure}[htbp]  
  \centering     
  \subfigure[original image]
  {
      \includegraphics[width=0.11\textwidth]{Correctness/10.png}
  }\hspace{-3mm}
  \subfigure[$8 \times 8$]
  {
      \includegraphics[width=0.11\textwidth]{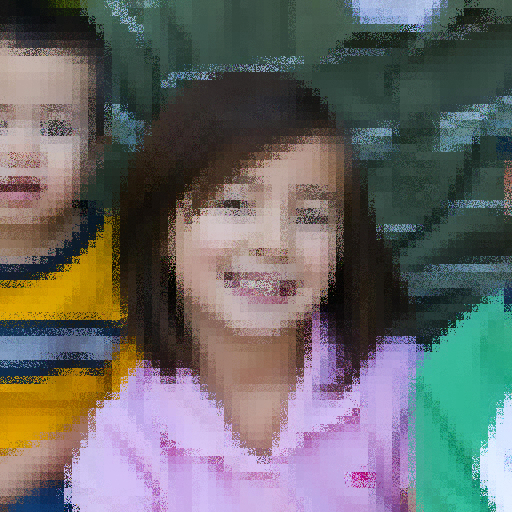}
  }\hspace{-3mm}
  \subfigure[$16 \times 16$]
  {
      
      \includegraphics[width=0.11\textwidth]{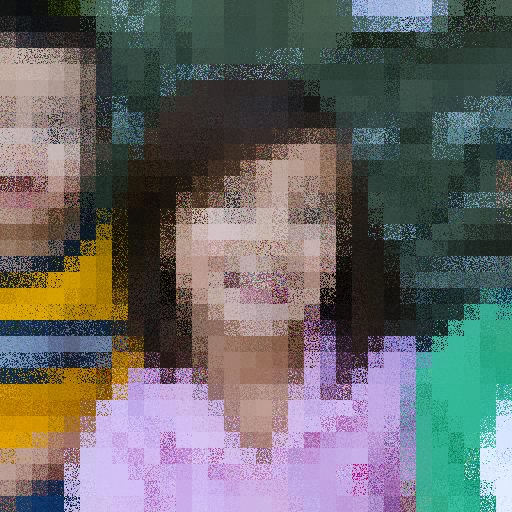}
  }\hspace{-3mm}
  \subfigure[$32 \times 32$]
  {
      
      \includegraphics[width=0.11\textwidth]{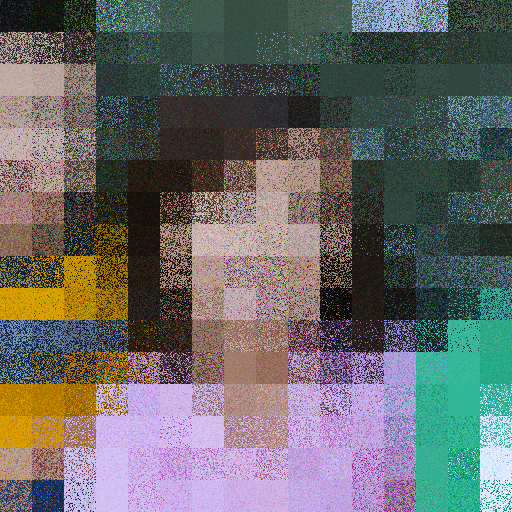}
  }\hspace{-3mm}
  \caption{The ciphertext images under the MFTPE framework preserving the sum and the range}
  \label{range}
\end{figure}

\begin{figure}[htbp]  
  \centering     
  \subfigure[Original image]
  {
      \includegraphics[width=0.11\textwidth]{Correctness/10.png}
  }\hspace{-3mm}
  \subfigure[$8 \times 8$]
  {
      \includegraphics[width=0.11\textwidth]{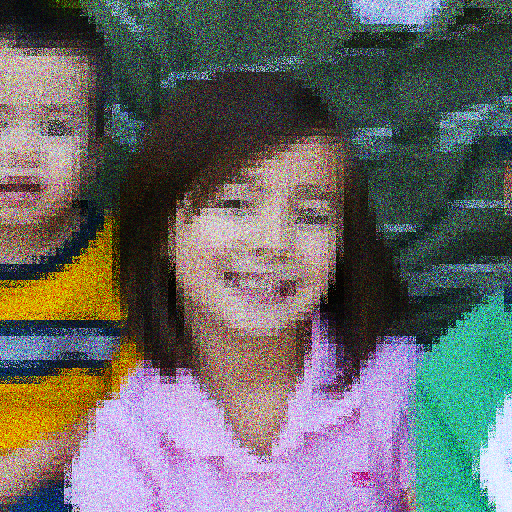}
  }\hspace{-3mm}
  \subfigure[$16 \times 16$]
  {
      
      \includegraphics[width=0.11\textwidth]{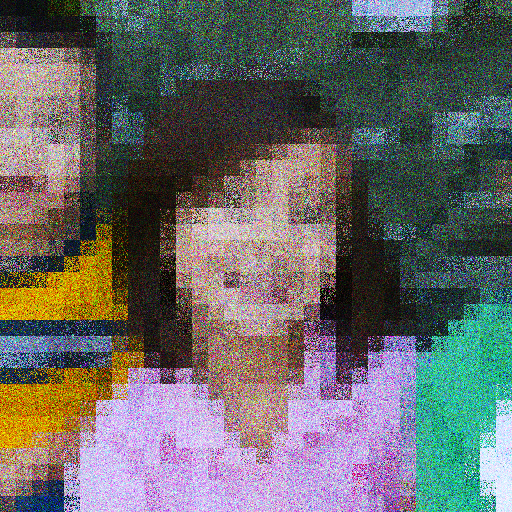}
  }\hspace{-3mm}
  \subfigure[$32 \times 32$]
  {
      
      \includegraphics[width=0.11\textwidth]{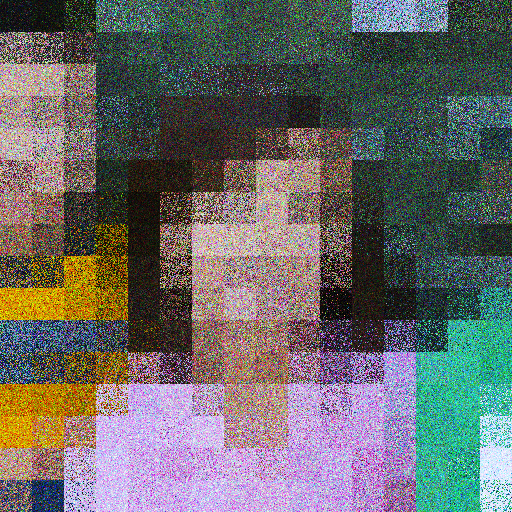}
  }\hspace{-3mm}
  \caption{The ciphertext images under the MFTPE framework preserving the sum and the weighted mean}
  \label{weight}
\end{figure}
\section{Experiment}\label{experimnet}
In this section, we perform sufficient experiments to validate the feasibility and superiority of the proposed scheme. The used dateset is the first 500 images from the Helen dataset. For fairness, it were preprocessed into PNG format and resized to $512\times 512$ dimensions. 

\begin{figure*}[!htbp]    
  \centering            
  \subfigure[R channel]   
  {   
      \includegraphics[width=0.3\textwidth]{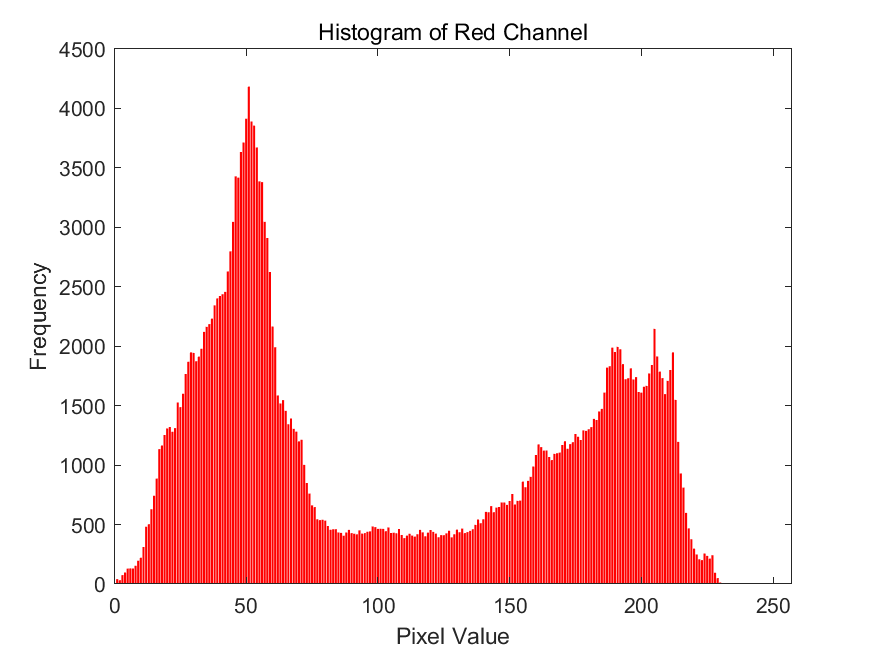}      
  }
  \subfigure[G channel]
  {
      
      \includegraphics[width=0.3\textwidth]{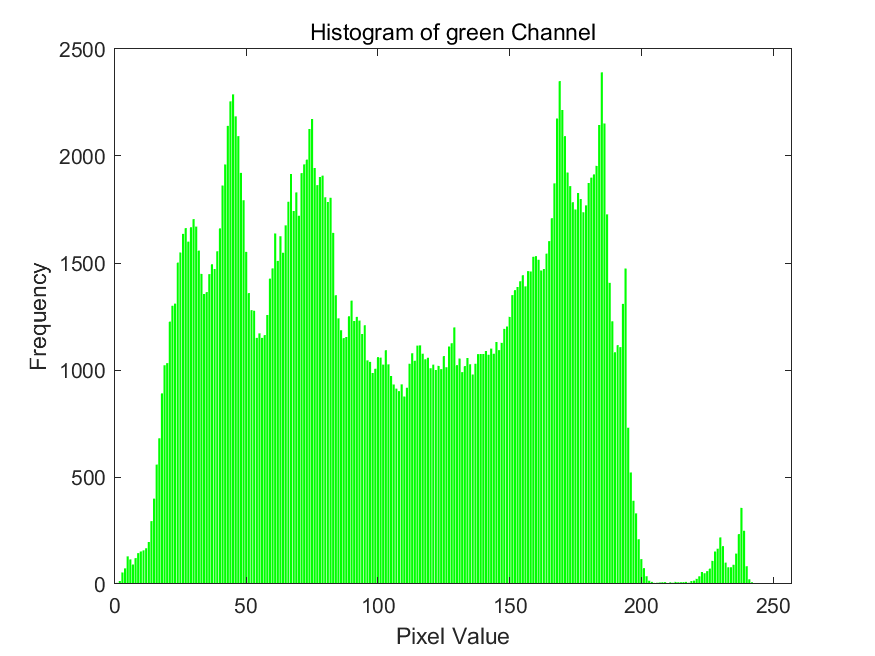}
  }
  \subfigure[B channel]
  {
      
      \includegraphics[width=0.3\textwidth]{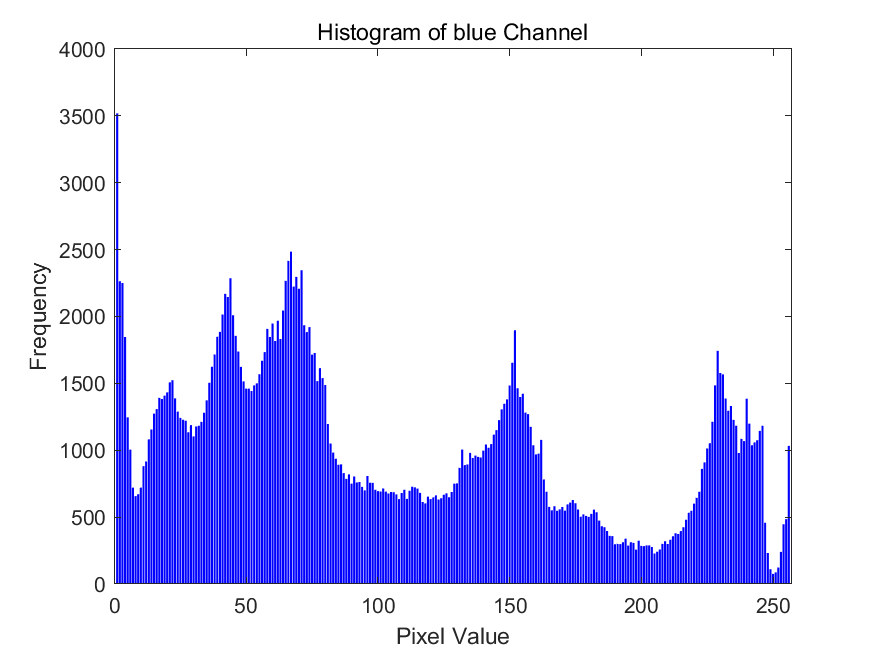}
  }
  \caption{Histograms of the plaintext image} 
  \label{ori RGB} 
\end{figure*}

\begin{figure*}[htbp]  
  \centering     
  \subfigure[R channel]
  {
      \includegraphics[width=0.3\textwidth]{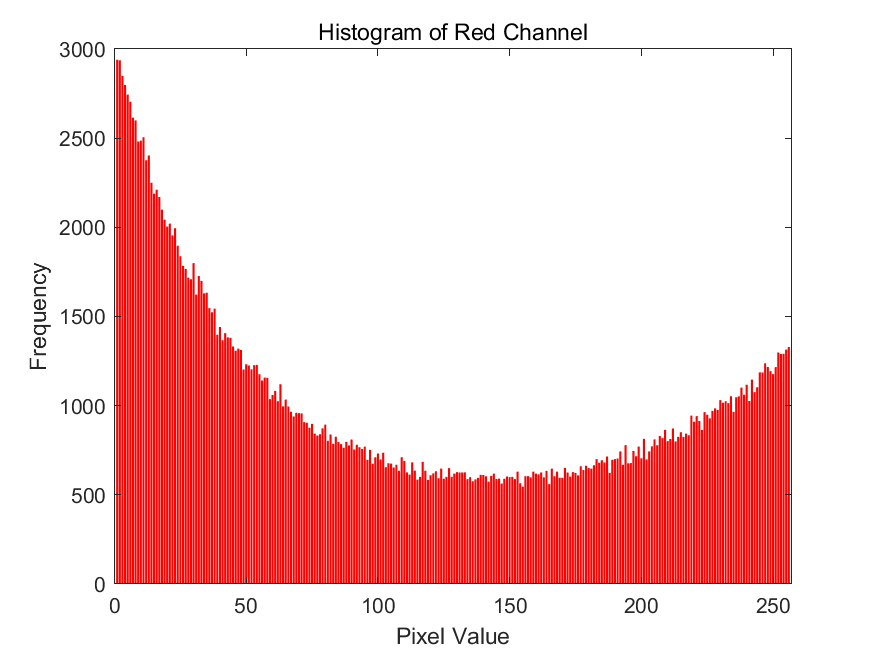}
  }
  \subfigure[G channel]
  {
      
      \includegraphics[width=0.3\textwidth]{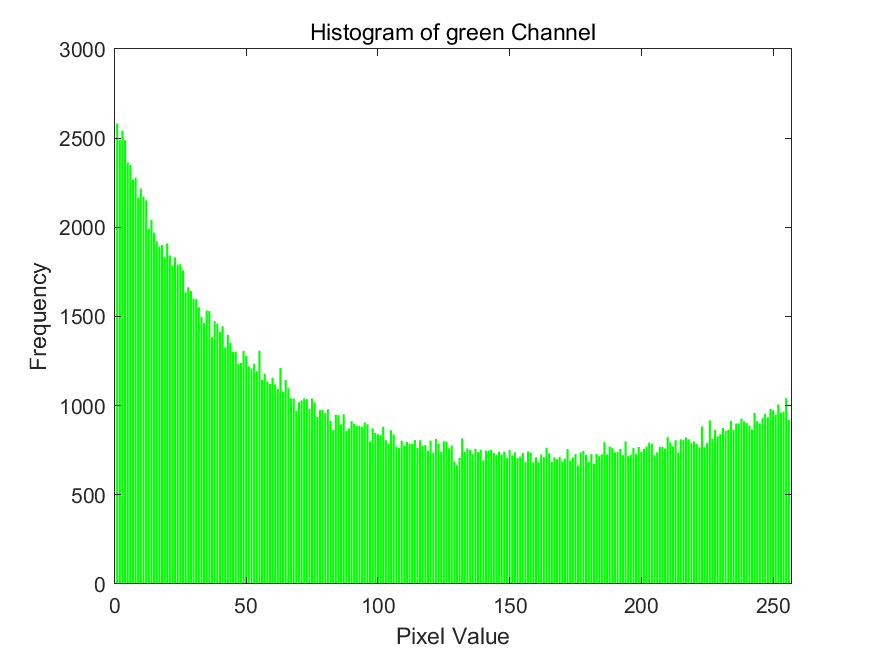}
  }
  \subfigure[B channel]
  {
      
      \includegraphics[width=0.3\textwidth]{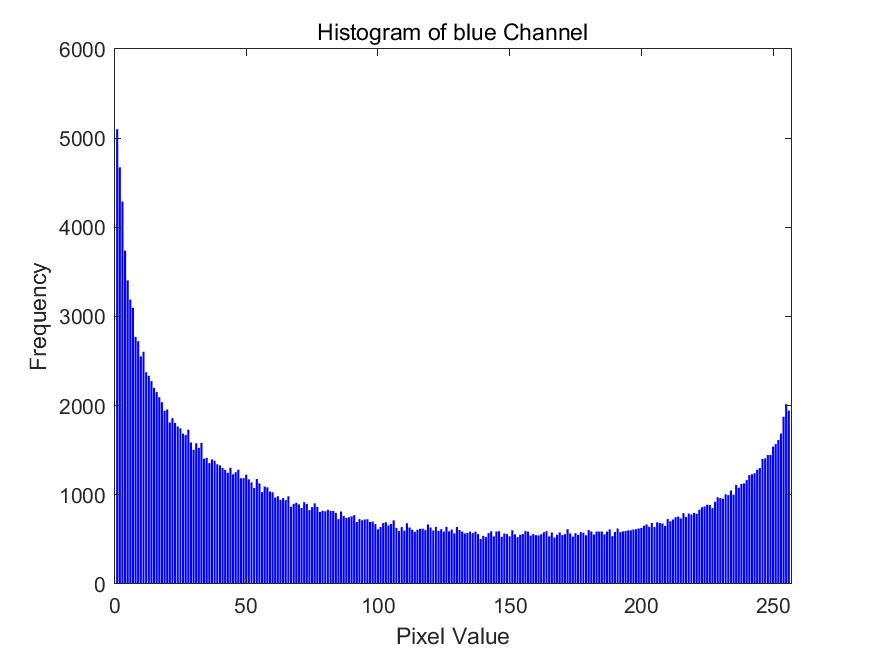}
  }
  \caption{Histograms of the ciphertext image when performing 25 rounds encryption}
  \label{25 RGB}
\end{figure*}

\begin{figure*}[htbp]  
  \centering     
  \subfigure[R channel]
  {
      \includegraphics[width=0.3\textwidth]{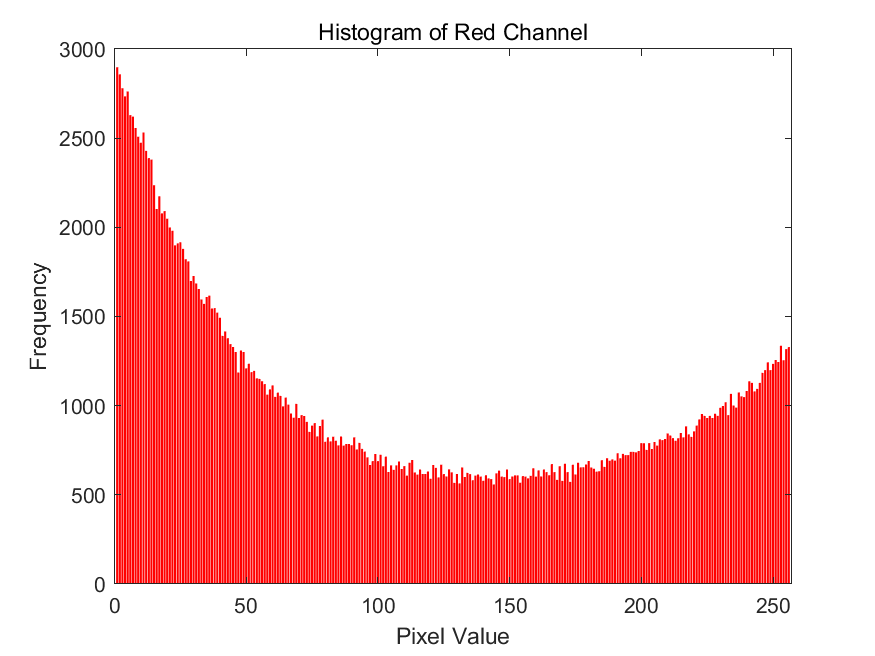}
  }
  \subfigure[G channel]
  {
      
      \includegraphics[width=0.3\textwidth]{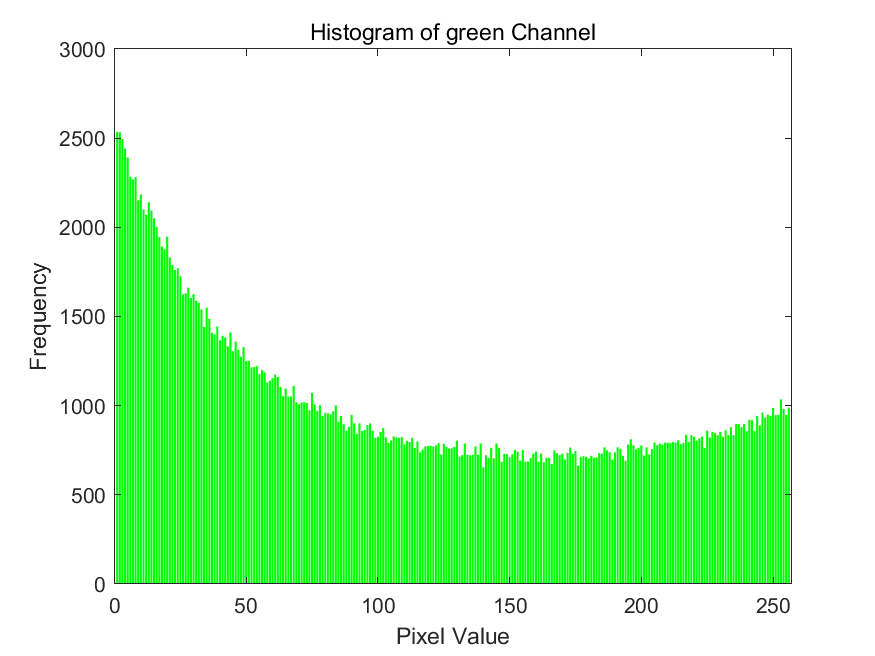}
  }
  \subfigure[B channel]
  {
      
      \includegraphics[width=0.3\textwidth]{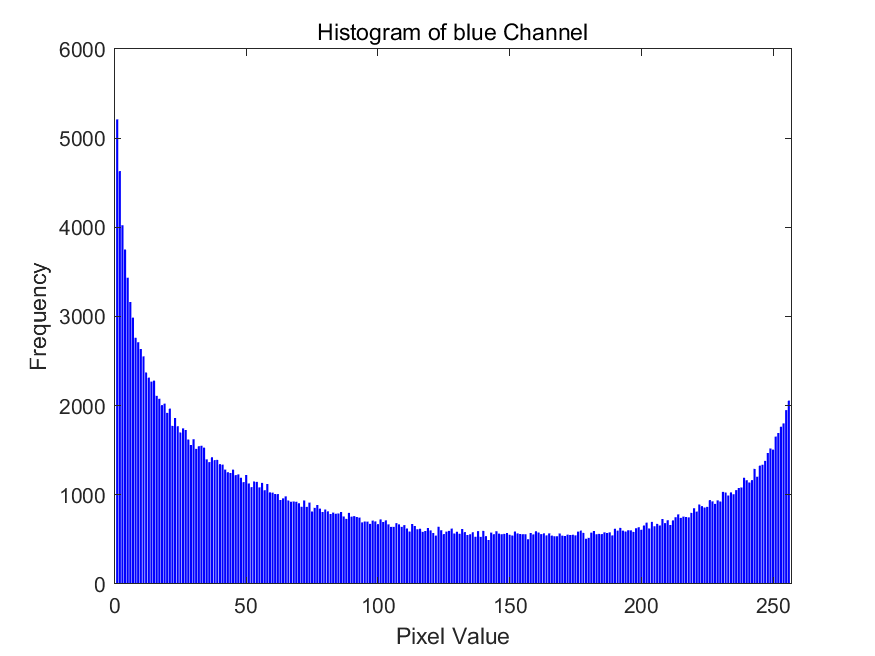}
  }
  \caption{Histograms of the ciphertext image when performing 100 rounds encryption}
  \label{100 RGB}
\end{figure*}

\begin{table}[!tb]
\centering
\caption{Encryption time of the MFTPE scheme preserving the geometric mean and the sum ($ms$)}
\label{Encryption time of geometric mean scheme}
\begin{tabular}{ccccc}
\hline
Block size  & $8\times8$ & $16\times16$ & $32\times32$ & $64\times64$ \\ \hline
Average time & 337.7 & 341  & 354.2  & 336.8   \\ 
Maximum time & 481 & 487  & 497  & 374   \\ 
Minimum time & 232 & 266  & 236  & 282   \\ \hline
\end{tabular}
\end{table}

\begin{table}[!tb]
\centering
\caption{Encryption time of the MFTPE scheme preserving the range and the sum 
 ($ms$)}
\label{Encryption time of Range scheme}
\begin{tabular}{ccccc}
\hline
Block size  & $8\times8$ & $16\times16$ & $32\times32$ & $64\times64$ \\ \hline
Average time & 166 & 171.8  & 159.6  & 159.4  \\ 
Maximum time & 172 & 174  & 173  & 177  \\ 
Minimum time & 157 & 171  & 156  & 155  \\ \hline
\end{tabular}
\end{table}

\begin{table}[!tb]
\centering
\caption{Encryption time of the MFTPE scheme preserving the sum and the weighted mean ($ms$)}
\label{Encryption time of Weighted mean scheme}
\begin{tabular}{ccccc}
\hline
 Block size  & $8\times8$ & $16\times16$ & $32\times32$ & $64\times64$ \\ \hline
Average time & 324.4 & 330  & 341.8  & 322.8  \\ 
Maximum time & 390 & 407  & 364  & 424  \\ 
Minimum time & 265 & 221  & 282  & 263  \\ \hline
\end{tabular}
\end{table}

\begin{table}
\centering
\caption{Encryption time of Tajik \textit{et al.} scheme ($ms$)}
\label{Encryption time for the tajik}
\begin{tabular}{ccccc}
\hline
 Block size  & $8\times8$ & $16\times16$ & $32\times32$ & $64\times64$ \\ \hline
Average time & 128.2 & 137.6  & 131.2  & 128  \\ 
Maximum time & 141 & 157  & 141  & 140  \\ 
Minimum time & 126 & 125  & 124  & 125  \\ \hline
\end{tabular}
\end{table}

\subsection{Correctness}

Fig. \ref{product} shows the original image and the ciphertext images using the proposed MFTPE framework with different size of blocks. The two factors, i.e., the sum and the geometric mean of pixels within a block, are the same between the plaintext and the ciphertexts. 
Similarly, Fig. \ref{range} shows the ciphertext images
when preserving the sum and the range of pixels within blocks, and Fig. \ref{weight} indicates the results when preserving the sum and the weighted mean.
The results indicate that the proposed MFTPE framework can work well, and the encryption effect becomes more random as the block size increases.

\subsection{Encryption Time}

The encryption time directly reflects the efficiency of the schemes, making it a crucial aspect in practical applications. In this subsection, we will compare the encryption times of three different MFTPE schemes and one TPE scheme. 
Table \ref{Encryption time of geometric mean scheme}, \ref{Encryption time of Range scheme} and \ref{Encryption time of Weighted mean scheme} show the one round encryption time of three MFTPE schemes, and Table \ref{Encryption time for the tajik} shows the one round encryption time of Tajik \textit{et al.}'s scheme \cite{tajik2019balancing}. 
The results indicate that the MFTPE scheme preserving the sum and the range and Tajik \textit{et al.}'s scheme have similar encryption time. 
We optimized the encryption process for the other two schemes (Table \ref{Encryption time of geometric mean scheme} and Table \ref{Encryption time of Weighted mean scheme}) by precomputing $\Phi_{d}^{-1}(s,p,n)$ and $\Phi_{d}^{-1}(s,w,n)$ for all pixel groups and storing them in advance. 
However, the results also indicate that the other two MFTPE schemes require more time than those of Tajik \textit{et al.}'s scheme.

\begin{figure*}[htbp]  
  \centering     
  \subfigure[Plaintext image]
  {
      \includegraphics[width=0.48\textwidth]{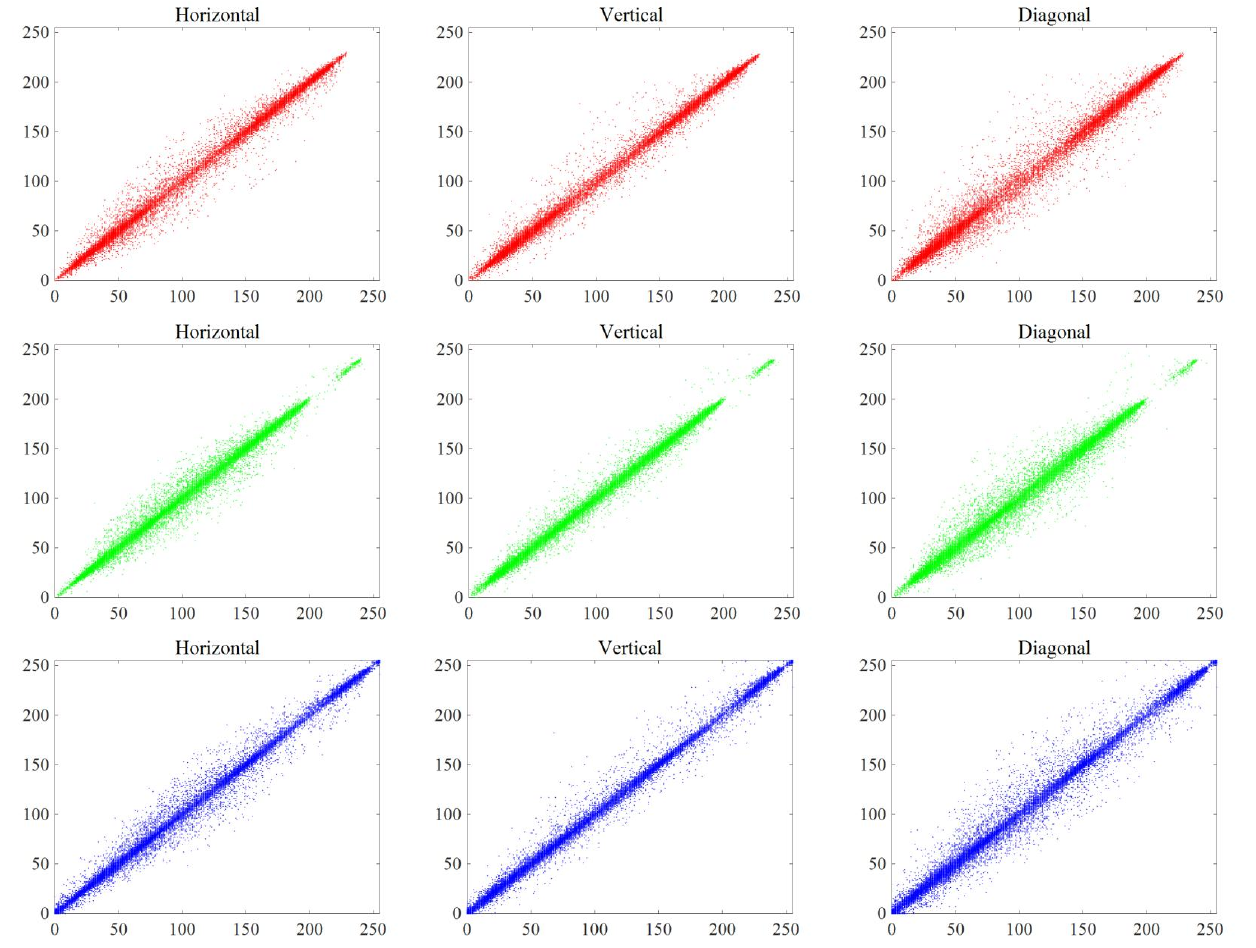}
  }\hspace{-3.3mm}
  \subfigure[Ciphertext image]
  {
      \includegraphics[width=0.48\textwidth]{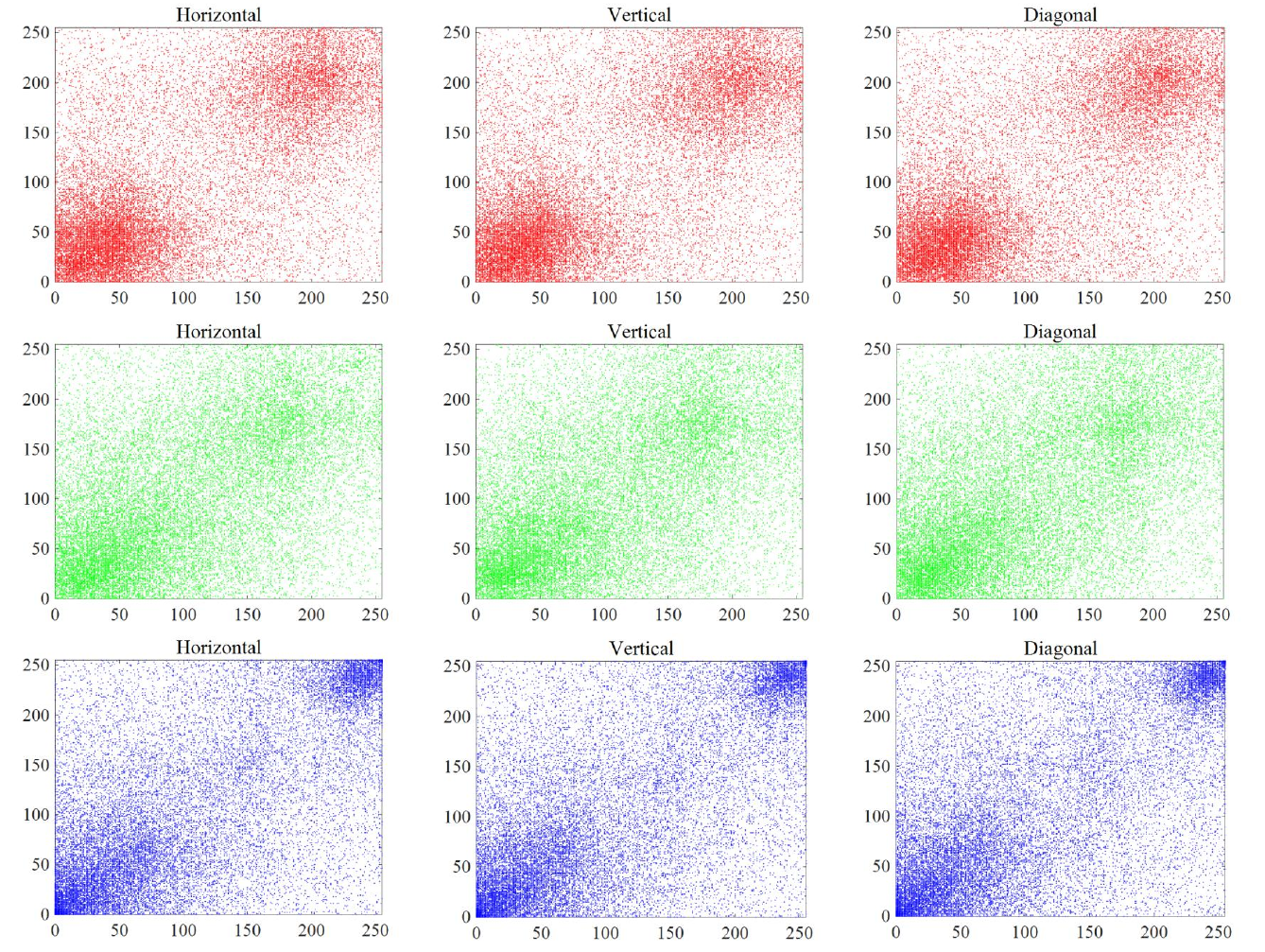}
  }\hspace{-3.3mm}
  \caption{Correlation of adjacent pixels for the plaintext image and the ciphertext image}
  \label{adjacent pixels analysis}
\end{figure*}

\subsection{Histogram analysis}

The histogram is a commonly used index to evaluate the effect of image encryption schemes.
For a secure TPE scheme, after multiple rounds of encryption, the image's histogram should tend to stabilize. 
Fig. \ref{ori RGB} illustrates the histograms of the RGB channels of the plaintext image. And Figs. \ref{25 RGB} and \ref{100 RGB} depict the histograms of the RGB channels of the ciphertext image encrypted by the MFTPE scheme preserving the sum and the weighted mean.

\subsection{Correlations}

The correlation of adjacent pixels is a frequently-used indicator for evaluating image encryption algorithms.
The correlations of the plaintext image and the ciphertext image in horizontal, vertical and diagonal directions are showed in Fig. \ref{adjacent pixels analysis}.
Note that the encryption scheme is the proposed MFTPE scheme, in which the sum and the weighted mean of pixels are unchanging before and after encryption.
The experimental results show that the correlations in the three directions of the plaintext image is very strong.
However, the correlations of the ciphertext image significantly decreases, providing robustness against statistical attacks.

\begin{figure}[htbp]
    \centering
    \includegraphics[width=0.45\textwidth]{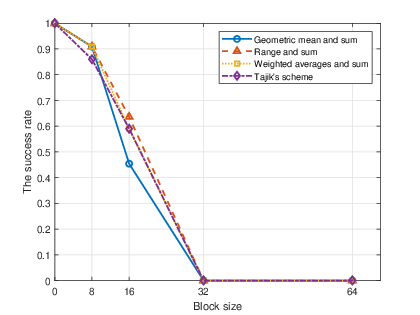}
    \caption{The success rate of face detection on the ciphertext images.}
    \label{face dect}
\end{figure}

\begin{figure}[htbp]  
  \centering     
  \subfigure
  {
      \includegraphics[width=0.5\textwidth]{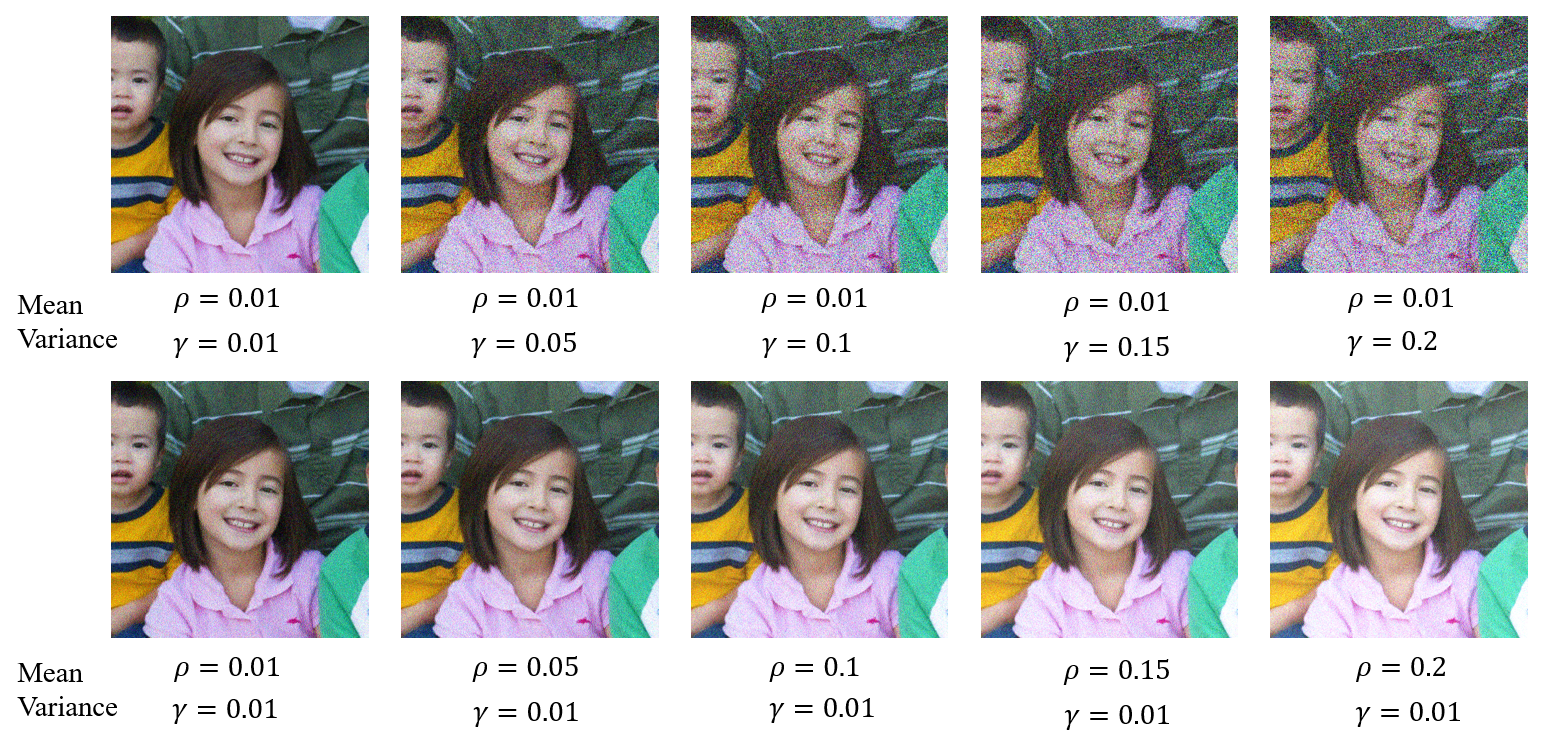}
  }
  \caption{Decrypted images when adding Gaussian noise}
  \label{gaussian}
\end{figure}

\begin{figure}[htbp]  
  \centering     
  \subfigure
  {
      \includegraphics[width=0.5\textwidth]{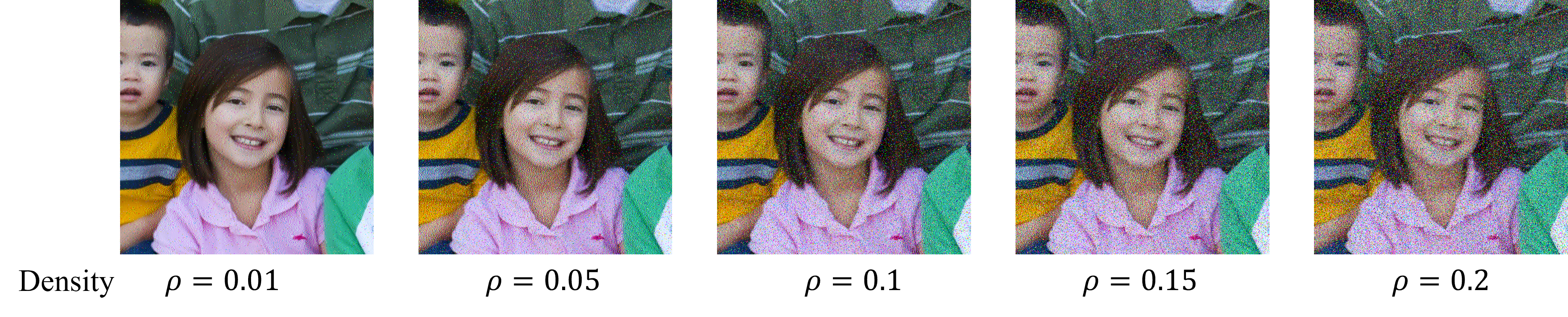}
  }
  \caption{Decrypted images when adding salt-and-pepper noise}
  \label{salt-and-pepper}
\end{figure}

\begin{figure}[htbp]  
  \centering     
  \subfigure
  {
      \includegraphics[width=0.5\textwidth]{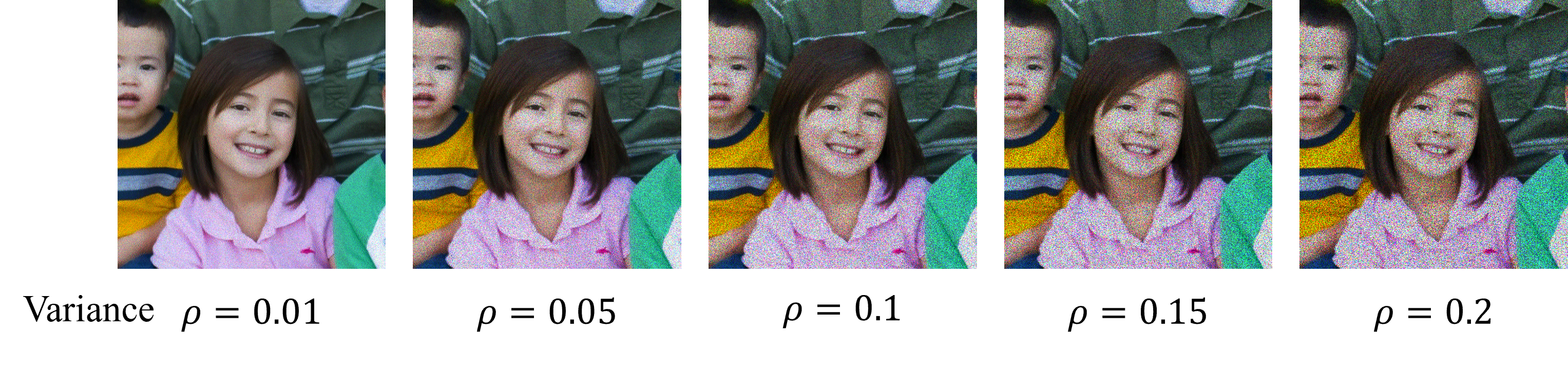}
  }
  \caption{Decrypted images when adding multiplicative noise}
  \label{multiplicative}
\end{figure}

\begin{figure*}[!htbp]  
  \centering     
  \subfigure[]
  {
      \includegraphics[width=0.15\textwidth]{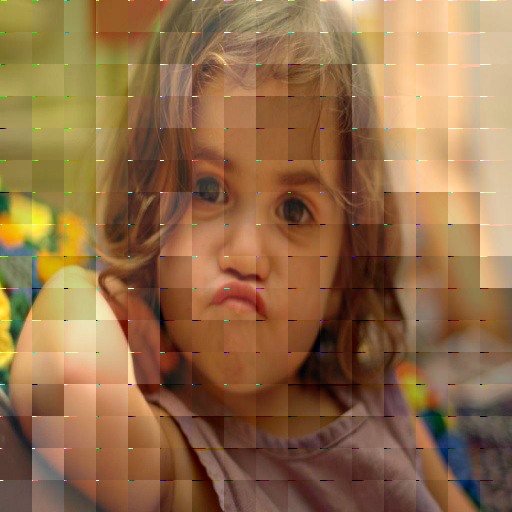}
  }\hspace{-3mm}
  \subfigure[Thumbnail of (a)]
  {
      \includegraphics[width=0.15\textwidth]{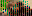}
  }\hspace{-3mm}
  \subfigure[Thumbnail of (a)]
  {
      
      \includegraphics[width=0.15\textwidth]{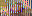}
  }\hspace{-3mm}
  \subfigure[]
  {
      
      \includegraphics[width=0.15\textwidth]{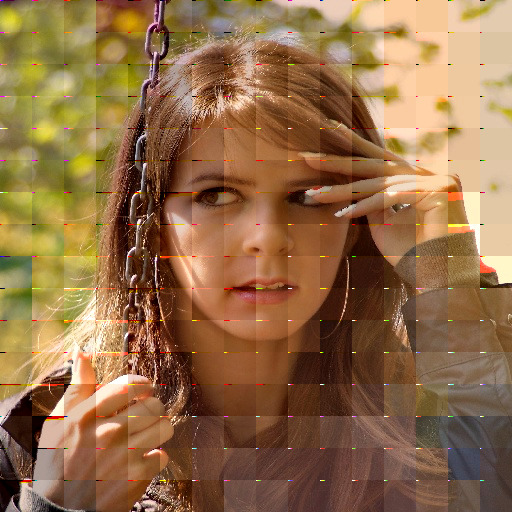}
  }\hspace{-3mm}
  \subfigure[Thumbnail of (d)]
  {
      \includegraphics[width=0.15\textwidth]{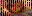}
  }\hspace{-3mm}
  \subfigure[Thumbnail of (d)]
  {
      
      \includegraphics[width=0.15\textwidth]{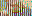}
  }\hspace{-3mm}
  \caption{Thumbnails of two different images}
  \label{Thumbnail comparison}
\end{figure*}
\subsection{Face detection attacks}

Here we examines the performance of face detection attacks on the encrypted images, focusing on the detection of faces due to their high privacy implications.
We utilize the widely used Face++ API \footnote{https://www.faceplusplus.com.cn/face-detection/} for this purpose, as Face++ is considered a state-of-the-art platform for facial detection. 

We conducted face detection on the encrypted images generated by the proposed three schemes and Tajik \textit{et al.}'s scheme \cite{tajik2019balancing}.
Fig. \ref{face dect} shows the success rates.
The results indicate that the proposed three schemes exhibit good resistance to face detection attacks when the block size is smaller than 32.
Additionally, the detection algorithms can not detect the faces in the ciphertext images generated by all the above-mentioned schemes when the block size is larger than 32.

\subsection{Noice attacks}
In most existing image encryption schemes, it is evidently impossible to perform the decryption process correctly if there exist some types of noise in the communication channel.
However, in the real word, noise cannot be avoided. 
To simulate this scenario, we consider the decryption effects when the encrypted ciphertext image are added some types of noise.
We employ three types of noise, i.e., Gaussian noise, salt-and-pepper noise, and multiplicative noise. Suppose that the encyrption scheme is MFTPE scheme preserving the sum and the range, as show in Section \ref{pix range}.

Fig. \ref{gaussian} shows the decrypted images when adding different types of Gaussian noise, Fig. \ref{salt-and-pepper} depicts the decrypted images with salt-and-pepper noise, and Fig. \ref{multiplicative} illustrates the decrypted image with multiplicative noise.
The experimental results indicate that the proposed MFTPE framework has the property of noise-resistant.

\section{Comparison}\label{comparison}

\subsection{Thumbnail}
Fig. \ref{same_thumbnail} shows that two different images can be formed the same thumbnail under the traditional TPE framework.
In this section, we will generate thumbnails for these two images under the proposed MFTPE framework.
Fig. \ref{Thumbnail comparison} (b) and (e) are the thumbnails using the MFTPE scheme preserving the sum and the geometric mean, and  Fig. \ref{Thumbnail comparison} (c) and (f) are the thumbnails
when preserving the sum and the range of pixels.
The results show that there is a significant difference between the thumbnails of these two plaintext images, making it easy to distinguish them from the thumbnails.

\begin{figure}
    \centering
    \includegraphics[width=0.4\textwidth]{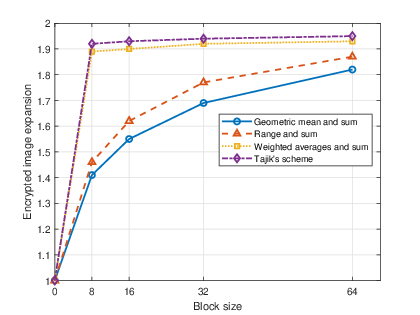}
    \caption{The storage expansion ratio}
    \label{image expansion ratio}
\end{figure}

\subsection{Storage expansion}
From the process of the MFTPE, the resolution and pixels of encrypted images are the same as those of plaintext images.
Therefore, the proposed encryption algorithm is a lossless encryption without data expansion.
However, the correlation and redundancy between pixels of images can be eliminated when storing these images.
We compare the required storage space between plaintext and ciphertext images. 
Here we compute an index, called storage expansion, denotes the storage space required for ciphertext divided by the storage space required for plaintext.

We compare the three specific schemes proposed in this paper with Tajik \textit{et al.}'s scheme \cite{tajik2019balancing}.
Fig. \ref{image expansion ratio} shows that as the size of blocks increases, the storage expansion increases. 
In addition, It can be seen from Fig. \ref{image expansion ratio} that the ciphertexts of Tajik \textit{et al.}'s scheme need more storage space compare to the proposed schemes.

\section{Conclusion}
\label{conclusion}
In this paper, we conduct an in-depth analysis of existing TPE schemes, and find a security issue
that two images may produce the same thumbnail, which we called ``thumbnail collision" .
We gradually derived a thumbnail collision probability from two blocks to multiple images. Additionally, we proposed three MFTPE schemes preserving multiple factors, significantly reducing the collision probabilities of thumbnails. To the best of our knowledge, this is the first approach aimed at reducing the collision probabilities of thumbnails. Furthermore, the experimental results demonstrate that the proposed method achieves high visual quality, robustly withstands face detection and noise attacks, and supports lossless decryption.

\ifCLASSOPTIONcaptionsoff
  \newpage
\fi

\bibliographystyle{IEEEtran}
\bibliography{cite}
\begin{IEEEbiography}[{\includegraphics[width=1in,height=1.25in,clip,keepaspectratio]{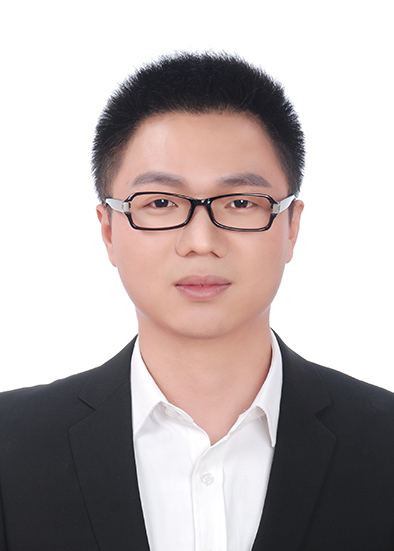}}]{Dong Xie}
received the Ph.D. degree in cryptography from Beijing University of Posts and Telecommunications, China, in 2017. He is currently an associate professor at the School of Computer and Information, Anhui Normal University. His research interests include cryptography, information security and image encryption. Dr. D. Xie is the co-author of over 40 scientific papers.
\end{IEEEbiography}
\begin{IEEEbiography}[{\includegraphics[width=1in,height=1.25in,clip,keepaspectratio]{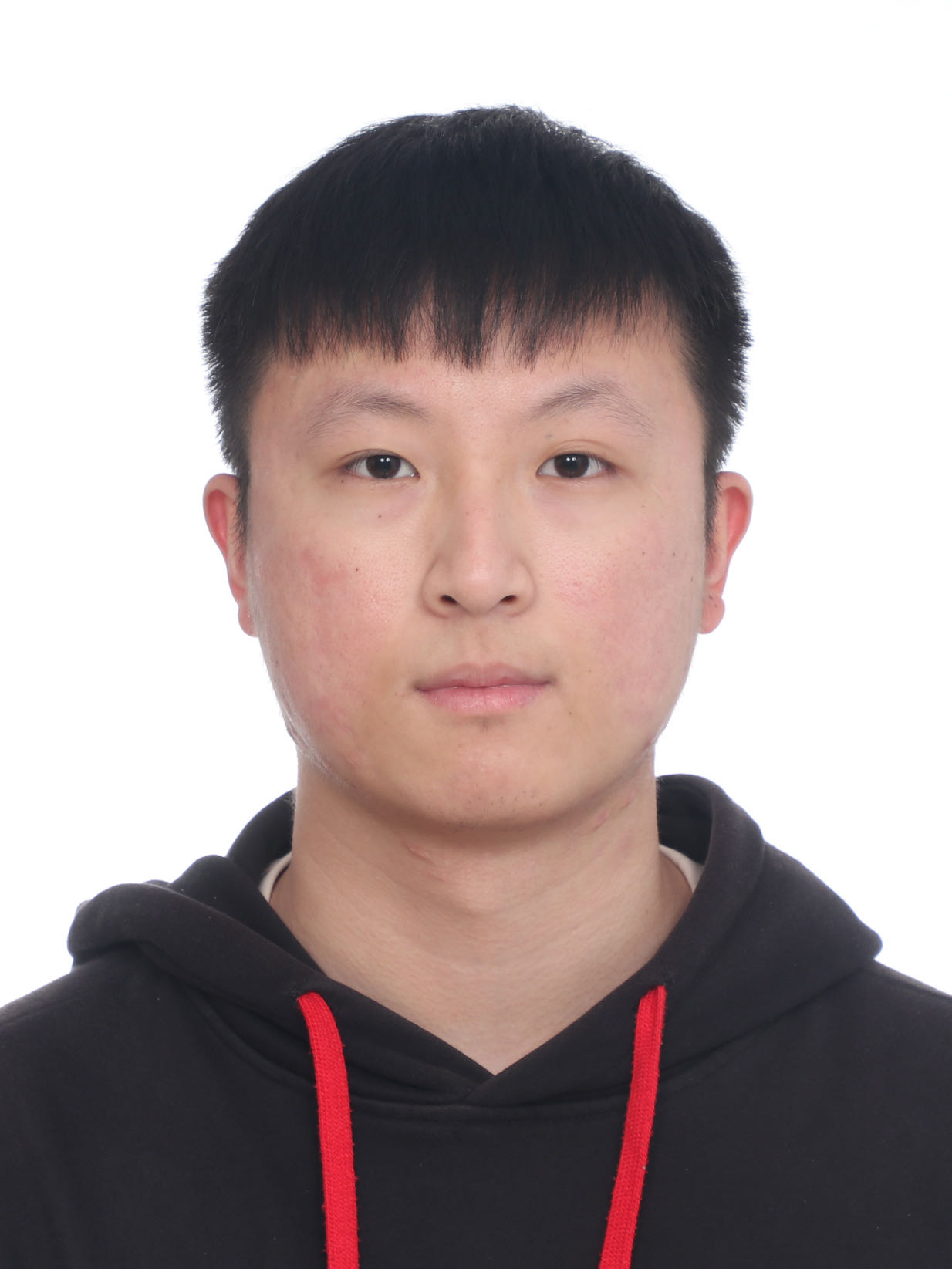}}]{Zhiyang Li}
is now a master student at the School of Computer and Information, Anhui Normal University. His research interests include image security.
\end{IEEEbiography}
\begin{IEEEbiography}[{\includegraphics[width=1in,height=1.25in,clip,keepaspectratio]{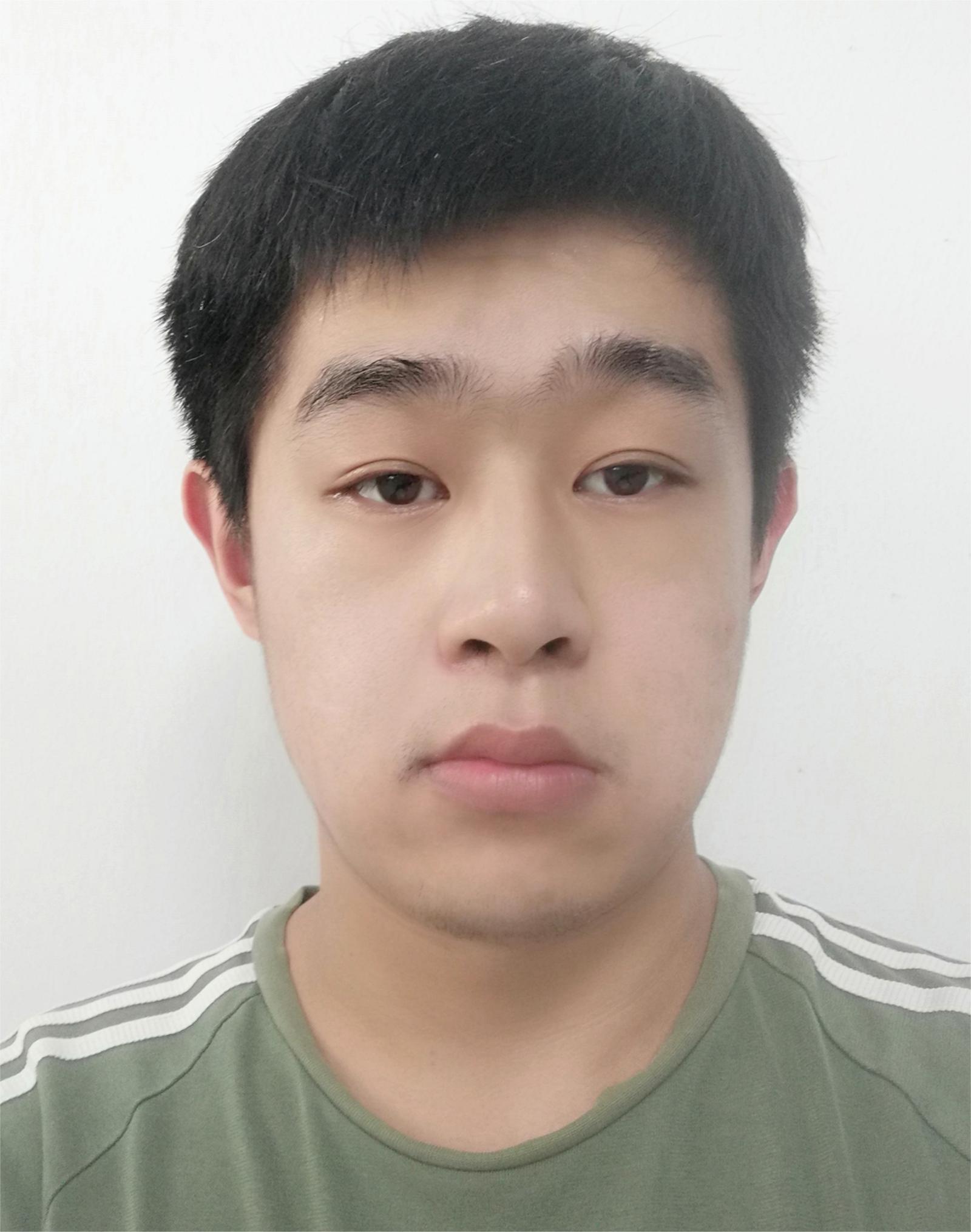}}]{Shuangxi Guo}
is now a master student at the School of Computer and Information, Anhui Normal University. His research interests include image encryption.
\end{IEEEbiography}
\begin{IEEEbiography}[{\includegraphics[width=1in,height=1.25in,clip,keepaspectratio]{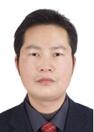}}]{Fulong Chen}
received his Bachelor's degree from Anhui Normal University in 2000, MS degree from China West Normal University in 2005, and PhD degree from Northwestern Polytechnical University in 2011. In the period from 2008 to 2010, he visited Rice University of USA and worked with Professor Walid Taha in the field of cyber-physical systems. Now he is a professor, master instructor and director of Department of Computer Science and Technology in Anhui Normal University. His research interests are embedded computing and pervasive computing, cyber-physical systems and high-performance computer architecture.
\end{IEEEbiography}
\begin{IEEEbiography}[{\includegraphics[width=1in,height=1.25in,clip,keepaspectratio]{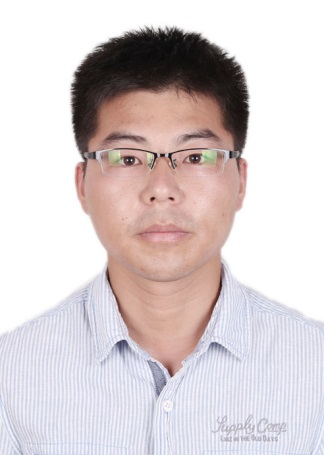}}]{Peng Hu}
received the Ph.D. degree in computer science from Nanjing University of Science and Technology, Nanjing, China, in 2021. He is currently an assistant professor with the School of Computer and Information, Anhui Normal University, Wuhu, China. His research interests include applied cryptography, information security, and Internet of Vehicles.
\end{IEEEbiography}

\end{document}